\documentclass[11pt,a4paper]{article}
\pdfoutput=1

\usepackage{amsfonts,amsmath,amsthm,color,enumitem,graphicx,url}
\usepackage[hmargin=30mm,vmargin=30mm]{geometry}
\usepackage[numbers,sort&compress]{natbib}
\usepackage{microtype}
\usepackage{hyperref}

\newcommand{\floor}[1]{\lfloor{#1}\rfloor}
\newcommand{\ceil}[1]{\lceil{#1}\rceil}
\newcommand{\myd}[2]{\lvert #1 #2 \rvert}
\newcommand{\s}[1]{{\lvert #1 \rvert}}

\newcommand{\eps}{\varepsilon}
\DeclareMathOperator{\poly}{poly}
\DeclareMathOperator{\diam}{diam}
\DeclareMathOperator{\interior}{int}

\newcommand{\myB}{\mathcal{B}} 
\newcommand{\myC}{\mathcal{C}} 
\newcommand{\myD}{\mathcal{D}} 
\newcommand{\myP}{\mathcal{P}} 
\newcommand{\myS}{\mathcal{S}}
\newcommand{\myT}{\mathcal{T}}
\newcommand{\myU}{\mathcal{U}}
\newcommand{\myV}{\mathcal{V}}

\newcommand{\myF}{\mathcal{F}}
\newcommand{\myG}{\mathcal{G}}
\newcommand{\myH}{\mathcal{H}}
\newcommand{\myR}{\mathcal{R}}

\newcommand{\myI}{\mathcal{I}}

\newcommand{\myA}{\mathcal{A}}

\newcommand{\MSTN}{\mathrm{MSFN}}
\newcommand{\MStTN}{\mathrm{MStFN}}
\newcommand{\optstab}{\mathrm{Stab}}
\newcommand{\optrel}{R^{*}}
\newcommand{\intW}{\interior(W\mspace{-1mu})}
\newcommand{\blue}{\text{blue}}
\newcommand{\red}{\text{red}}

\newcommand{\rstabt}{137/60}
\newcommand{\rhub}{\frac{137}{30}}

\hypersetup{
    colorlinks=true,
    linkcolor=black,
    citecolor=black,
    filecolor=black,
    urlcolor=[rgb]{0,0.1,0.5},
    pdftitle={Improved approximation algorithms for relay placement},
    pdfauthor={Alon Efrat, S\'andor P. Fekete, Joseph S. B. Mitchell, Valentin Polishchuk, Jukka Suomela},
}

\newtheorem{theorem}{Theorem}
\newtheorem{lemma}{Lemma}
\newtheorem{corollary}{Corollary}

\theoremstyle{definition}

\theoremstyle{remark}
\newtheorem{remark}{Remark}

\newenvironment{myabstract}
               {\list{}{\listparindent 1.5em%
                        \itemindent    \listparindent
                        \leftmargin    0cm
                        \rightmargin   0cm
                        \parsep        0pt}%
                \item\relax}
               {\endlist}

\newenvironment{mycover}
               {\list{}{\listparindent 0pt
                        \itemindent    \listparindent
                        \leftmargin    0cm
                        \rightmargin   0cm
                        \parsep        0pt}%
                \raggedright
                \item\relax}
               {\endlist}

\newcommand{\myauthor}[1]{#1\par\smallskip}
\newcommand{\myaff}[1]{{\small #1\par}\bigskip}

\setenumerate{label=(\arabic*)}
\newcommand{\mytitlebreak}{\texorpdfstring{\\}{ }}

\begin{document}

\mbox{}
\begin{mycover}
{\huge \bfseries Improved Approximation Algorithms \\ for Relay Placement \par}

\bigskip
\bigskip
\myauthor{ALON EFRAT}
\myaff{Department of Computer Science, University of Arizona}

\myauthor{S\'ANDOR P.\ FEKETE}
\myaff{Department of Computer Science, Braunschweig University of Technology}

\myauthor{JOSEPH S.\ B.\ MITCHELL}
\myaff{Department of Applied Mathematics and Statistics, Stony Brook University}

\myauthor{VALENTIN POLISHCHUK}
\myaff{Communications and Transport Systems, Link\"oping University
    \par\smallskip
    Helsinki Institute for Information Technology HIIT\\
    Department of Computer Science, University of Helsinki}

\myauthor{JUKKA SUOMELA}
\myaff{Helsinki Institute for Information Technology HIIT, \\
    Department of Computer Science, Aalto University}
\end{mycover}

\bigskip
\begin{myabstract}
\noindent\textbf{Abstract.}
In the relay placement problem the input is a set of sensors and a number $r \ge 1$, the communication range of a relay.  In the \emph{one-tier} version of the problem the objective is to place a minimum number of relays so that between every pair of sensors there is a path \emph{through sensors and/or relays} such that the consecutive vertices of the path are within distance $r$ if both vertices are relays and within distance~1 otherwise.  The \emph{two-tier} version adds the restrictions that the path must go \emph{through relays, and not through sensors}.  We present a $3.11$-approximation algorithm for the one-tier version and a PTAS for the two-tier version.  We also show that the one-tier version admits no PTAS, assuming P${}\ne{}$NP.
\end{myabstract}
\thispagestyle{empty}
\setcounter{page}{0}
\newpage


\section{Introduction}

A sensor network consists of a large number of low-cost autonomous devices, called \emph{sensors}. Communication between the sensors is performed by wireless radio with very limited range, e.g., via the Bluetooth protocol. To make the network connected, a number of additional devices, called \emph{relays}, must be judiciously placed within the sensor field.  Relays are typically more advanced and more expensive than sensors, and, in particular, have a larger communication range.  For instance, in addition to a Bluetooth chip, each relay may be equipped with a WLAN transceiver, enabling communication between distant relays.  The problem we study in this paper is that of placing a \emph{minimum number} of relays to ensure the connectivity of a sensor network.

Two models of communication have been considered in the literature \cite{%
    bredin10deploying,chen00approximations,chen01approximations,%
    cheng08relay,liu06optimal,lloyd07relay,srinivas06mobile,%
    zhang07fault-tolerant%
}.  In both models, a sensor and a relay can communicate if the distance between them is at most~1, and two relays can communicate if the distance between them is at most~$r$, where $r\ge1$ is a given number.  The models differ in whether direct communication between sensors is allowed.  In the \emph{one-tier} model two sensors can communicate if the distance between them is at most~1.  In the \emph{two-tier} model the sensors do not communicate at all, no matter how close they are.  In other words, in the two-tier model the sensors may only link to relays, but not to other sensors.

Formally, the input to the relay placement problem is a set of $n$ sensors, identified with their locations in the plane, and a number $r\ge1$, the communication range of a relay (by scaling, without loss of generality, the communication range of a sensor is~$1$).  The objective in the \emph{one-tier} relay placement is to place a minimum number of relays so that between every pair of sensors there exists a path, \emph{through sensors and/or relays}, such that the consecutive vertices of the path are within distance $r$ if both vertices are relays, and within distance~1 otherwise.  The objective in the \emph{two-tier} relay placement is to place a minimum number of relays so that between every pair of sensors there exists a path \emph{through relays} such that the consecutive vertices of the path are within distance $r$ if both vertices are relays, and within distance~1 if one of the vertices is a sensor and the other is a relay (going directly from a sensor to a sensor is forbidden).

\subsection{Previous Work}

One-tier relay placement in the special case of $r=1$ \cite{bredin10deploying,cheng08relay} is equivalent to finding a Steiner tree with minimum number of Steiner nodes and bounded edge length -- the problem that was studied under the names STP-MSPBEL \cite{lin99steiner}, SMT-MSPBEL \cite{lloyd07relay,zhang07fault-tolerant}, MSPT \cite{mandoiu00note}, and STP-MSP \cite{chen00approximations,chen01approximations,cheng08relay,liu06optimal,srinivas06mobile}.
\citet{lin99steiner} proved that the problem is NP-hard and gave a 5-approximation algorithm.  Chen et al.\ \cite{chen00approximations,chen01approximations} showed that the algorithm of Lin and Xue is actually a 4-approximation algorithm, and gave a 3-approximation algorithm; \citet{cheng08relay} gave a 3-approximation algorithm with an improved running time, and a randomised 2.5-approximation algorithm.  Chen et al.\ \cite{chen00approximations,chen01approximations} presented a polynomial-time approximation scheme (PTAS) for minimising the \emph{total} number of vertices in the tree (i.e., with the objective function being the number of the original points plus the number of Steiner vertices) for a restricted version of the problem, in which in the minimum spanning tree of the set the length of the longest edge is at most constant times the length of the shortest edge.

For the general case of \emph{arbitrary} $r\ge1$, the current best approximation ratio for one-tier relay placement is due to \citet{lloyd07relay}, who presented a simple 7-approximation algorithm, based on ``Steinerising'' the minimum spanning tree of the sensors.  In this paper we give an algorithm with an improved approximation ratio of 3.11.

Two-tiered relay placement (under the assumptions that the sensors are uniformly distributed in a given region and that $r\ge4$) was considered by \citet{hao04fault-tolerant} and \citet{tang06relay} who suggested constant-factor approximation algorithms for several versions of the problem.  \citet[Thm.~4.1]{lloyd07relay} and \citet[Thm.~1]{srinivas06mobile} developed a general framework whereby given an $\alpha$-approximate solution to Disk Cover (finding minimum number of unit disks to cover a given set of points) and a $\beta$-approximate solution to STP-MSPBEL (see above), one may find an approximate solution for the two-tier relay placement.  In more details, the algorithm in \citet{lloyd07relay} works for arbitrary $r\ge1$ and has an approximation factor of $2\alpha+\beta$; the algorithm in \citet{srinivas06mobile} works for $r\ge2$ and guarantees an $(\alpha+\beta)$-approximate solution.  Combined with the best known approximation factors for the Disk Cover \cite{hochbaum85approximation} and STP-MSPBEL \cite{chen00approximations,chen01approximations,cheng08relay}, these lead to $5+\eps$ and $4+\eps$ approximations for the relay placement respectively.  In this paper we present a PTAS for the two-tiered relay placement; the PTAS works directly for the relay placement, without combining solutions to other problems.

A different line of research \cite{misra08constrained,carmi07covering} concentrated on a ``discrete'' version of relay placement, in which the goal is to pick a minimum subset of relays from a \emph{given} set of possible relay locations.  In this paper we allow the relays to reside anywhere in the plane.

\subsection{Contributions}

We present new results on approximability of relay placement:
\begin{itemize}[itemsep=0.5ex]
    \item In Section~\ref{sec_apx1tier} we give a simple $O(n\log n)$-time 6.73-approximation algorithm for the one-tier version.
    \item In Section~\ref{sec_apx1tierim} we present a polynomial-time 3.11-approximation algorithm for the one-tier version.
    \item In Section~\ref{sec_inapx1tier} we show that there is no PTAS for one-tier relay placement (assuming that $r$ is part of the input, and P${}\ne{}$NP).
    \item In Section~\ref{sec_ptas} we give a PTAS for two-tier relay placement.
\end{itemize}
Note that the \emph{number} of relays in a solution may be exponential in the size of the input (number of bits).  Our algorithms produce a succinct representation of the solution.  The representation is given by a set of points and a set of line segments; the relays are placed on each point and equally-spaced along each segment.


\section{Blobs, Clouds, Stabs, Hubs, and Forests}\label{sec_prelim}

In this section we introduce the notions, central to the description of our algorithms for one-tier relay placement.  We also provide lower bounds.

\subsection{Blobs and Clouds}

We write $\myd{x}{y}$ for the Euclidean distance between $x$ and~$y$. Let $V$ be a given set of sensors (points in the plane).  We form a unit disk graph $\myG = (V,E)$ and a disk graph $\myF = (V,F)$ where
\begin{align*}
    E &= \bigl\{ \{u,v\} : \myd{u}{v} \le 1 \bigr\}, \\
    F &= \bigl\{ \{u,v\} : \myd{u}{v} \le 2 \bigr\};
\end{align*}
see Figure~\ref{fig:clouds}a.

\begin{figure}[t]
    \centering
    \scalebox{0.9}{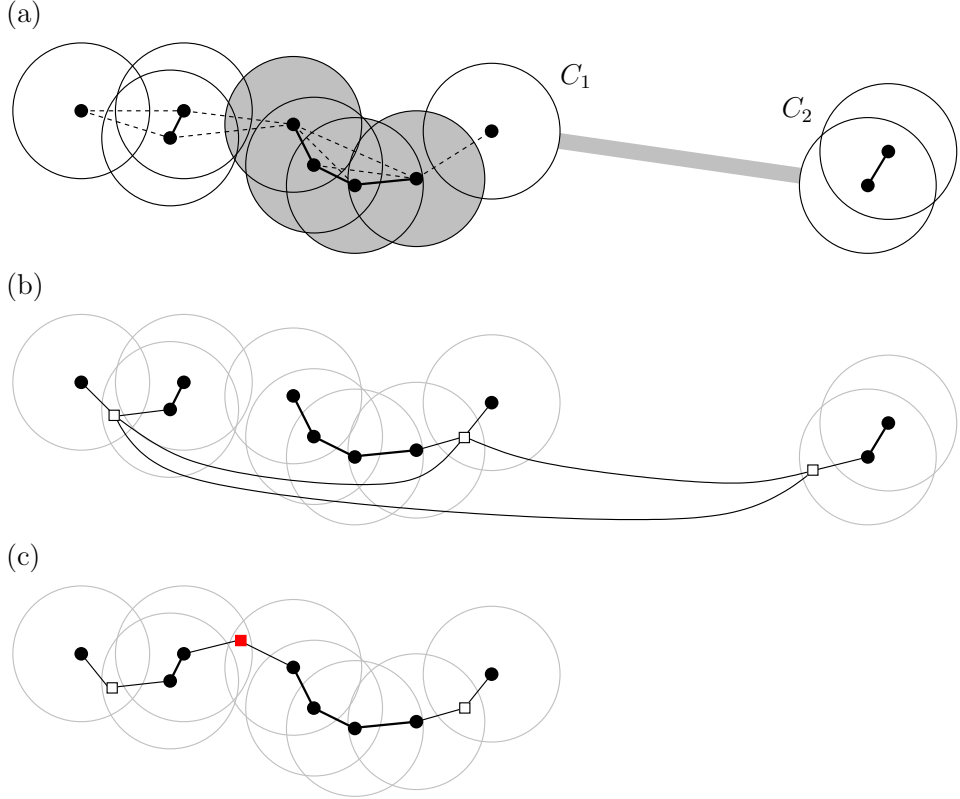}
    \caption{(a)~Dots are sensors in $V$, solid lines are edges in $E$ and $F$, and dashed lines are edges in $F$ only. There are 5 blobs in $\myB$ (one of them highlighted) and 2 clouds $C_1, C_2 \in \myC$. The wide grey line is the only edge in $\MStTN(\myC)$, which happens to be equal to $\MSTN(\myC)$ here. (b)~Stabs. (c)~Hubs.}\label{fig:clouds}
\end{figure}

A \emph{blob} is defined to be the union of the unit disks centered at the sensors that belong to the same connected component of $\myG$. We use $B$ to refer to a blob, and $\myB$ for the set of all blobs.

Analogously, a \emph{cloud} $C \in \myC$ is the union of the unit disks centered at the sensors that belong to the connected component of the graph $\myF$. The sensors in a blob can communicate with each other without relays, while the ones in a cloud might not, even though their disks may overlap. Each cloud $C \in \myC$ consists of one or more blobs $B \in \myB$; we use $\myB_C$ to denote the blobs that form the cloud~$C$.

\subsection{Stabs and Hubs}

A \emph{stab} is a relay with an infinite communication range ($r=\infty$). A \emph{hub} is a relay without the ability to communicate with the other relays (thus hubs can enable communication within one cloud, but are of no use in communicating between clouds). As we shall see, a solution to stab or hub placement can be used as the first step towards a solution for relay placement.

If we are placing stabs, it is necessary and sufficient to have a stab in each blob to ensure communication between all sensors (to avoid trivialities we assume there is more than one blob). Thus, stab placement is a special case of the set cover problem: the universe is the blobs, and the subsets are sets of blobs that have a point in common.  We use $\optstab(\myB')$ to denote the minimum set of stabs that stab each blob in $\myB' \subseteq \myB$. In the example in Figure~\ref{fig:clouds}b small rectangles show an optimal solution to the stab placement problem; 3 stabs are enough.

If we are placing hubs, it is necessary (assuming more than one blob in the cloud), but not sufficient, to have a hub in each blob to ensure communication between sensors within one cloud.  In fact, hub placement can be interpreted as a special case of the \emph{connected} set cover problem \cite{cerdeira05requiring,shuai06connected}.  In the example in Figure~\ref{fig:clouds}c small rectangles show an optimal solution to the hub placement problem for the cloud $C = C_1$; in this particular case, 2 stabs within the cloud $C$ were sufficient to ``pierce'' each blob in $\myB_C$ (see Figure~\ref{fig:clouds}b), however, an additional hub (marked red in Figure~\ref{fig:clouds}c) is required to ``stitch'' the blobs together (i.e., to establish communication between the blobs). The next lemma shows that, in general, the number of additional hubs needed  is less than the number of stabs:

\begin{lemma}\label{lem_stab2hub}
    Given a feasible solution\/ $S$ to stab placement on\/ $\myB_C$, we can obtain in polynomial time a feasible solution to hub placement on\/ $\myB_C$ with\/ $2 \s{S}-1$ hubs.
\end{lemma}

\begin{proof}
    Let $\myH$ be the graph, whose nodes are the sensors in the cloud $C$ and the stabs in $S$, and whose edges connect two devices if either they are within distance~1 from each other or if both devices are stabs (i.e., there is an edge between \emph{every} pair of the stabs).  Switch off communication between the stabs, thus turning them into hubs.  Suppose that this breaks $\myH$ into $k$ connected components. There must be a stab in each connected component. Thus, $\s{S} \ge k$.

    If $k > 1$, by the definition of a cloud, there must exist a point where a unit disk covers at least two sensors from two different connected components of $\myH$. Placing a hub at the point decreases the number of the connected components by at least~1. Thus, after putting at most $k-1$ additional hubs, all connected components will merge into one.
\end{proof}

\subsection{Steiner Forests and Spanning Forests with Neighbourhoods}

Let $\myP$ be a collection of planar subsets; call them \emph{neighbourhoods}. (In Section~\ref{sec_apx1tier} the neighbourhoods will be the clouds, in Section~\ref{sec_apx1tierim} they will be ``clusters'' of clouds.)  For a plane graph $G$, let $\myG_{\myP}=(\myP, E(G))$ be the graph whose vertices are the neighbourhoods and two neighbourhoods $P_1,P_2\in\myP$ are adjacent whenever $G$ has a vertex in $P_1$, a vertex in $P_2$, and a path between the vertices.

The \emph{Minimum Steiner Forest with Neighbourhoods} on $\myP$, denoted $\MStTN(\myP)$, is a \emph{minimum-length} plane graph $G$ such that $\myG_{\myP}=(\myP, E(G))$ is connected. The $\MStTN$ is a generalisation of the Steiner tree of a set of points. Note that $\MStTN$ is slightly different from Steiner tree with neighbourhoods (see, e.g., \citet{yang07minimum}) in that we are only counting the part of the graph \emph{outside} $\myP$ towards its length (since it is not necessary to connect neighbourhoods beyond their boundaries).

Consider a complete weighted graph whose vertices are the neighbourhoods in $\myP$ and whose edge weights are the distances between them.  A minimum spanning tree in the graph is called the \emph{Minimum Spanning Forest with Neighbourhoods} on $\myP$, denoted $\MSTN(\myP)$.  A natural embedding of the edges of the forest is by the straight-line segments that connect the corresponding neighbourhoods; we will identify $\MSTN(\myP)$ with the embedding.  (As with $\MStTN$, we count the length of $\MSTN$ only \emph{outside}~$\myP$.)

We denote by $\s{\MStTN(\myP)}$ and $\s{\MSTN(\myP)}$ the total length of the edges of the forests.  It is known that
\[
    \s{\MSTN({P})} \,\le\, \frac{2}{\sqrt{3}} \s{\MStTN({P})}
\]
for a \emph{point} set $P$, where $2/\sqrt{3}$ is the \emph{Steiner ratio} \cite{du90approach}. The following lemma generalises this to neighbourhoods.
\begin{lemma}\label{lem_StRatio}
    For any\/ $\myP$, $\s{\MSTN(\myP)} \le (2/\sqrt{3}) \s{\MStTN(\myP)}$.
\end{lemma}
\begin{proof}
    If $\myP$ is erased, $\MStTN(\myP)$ falls off into a forest, each tree of which is a minimum Steiner tree on its leaves; its length is within the Steiner ratio of minimum spanning tree length.
\end{proof}

\subsection{Lower Bounds on the Number of Relays}\label{ssec:lower-bounds}

Let $\optrel$ be an optimal set of relays.  Let $\myR$ be the communication graph on the relays $\optrel$ alone, i.e., without sensors taken into account; two relays are connected by an edge in $\myR$ if and only if they are within distance $r$ from each other. Suppose that $\myR$ is embedded in the plane with vertices at relays and line segments joining communicating relays.  The embedding spans all clouds, for otherwise the sensors in a cloud would not be connected to the others. Thus, in $\myR$ there exists a forest $\myR'$, whose embedding also spans all clouds. Let $\s{\myR'}$ denote the total length of the edges in $\myR'$. By definition of $\MStTN(\myC)$, we have $\s{\myR'} \ge \s{\MStTN(\myC)}$.

Let $m$, $v$, and $k$ be the number of edges, vertices, and trees of $\myR'$.  Since each edge of $\myR'$ has length at most $r$, we have $\s{\myR'} \le m r = (v - k) r$. Since $v \le \s{\optrel}$, since there must be a relay in every blob and every cloud, and since the clouds are disjoint, it follows that
\begin{align}
    \s{\optrel} &\ge \s{\MStTN(\myC)} / r, \label{eq_lower_bound_MStTN}\\
    \s{\optrel} &\ge \s{\optstab(\myB)}, \label{eq_lower_bound_Stab}\\
    \s{\optrel} &\ge \s{\myC}. \label{eq_lower_bound_C}
\end{align}


\section{A 6.73-Approximation Algorithm for\mytitlebreak One-Tier Relay Placement}\label{sec_apx1tier}

In this section we give a simple 6.73-approximation algorithm for relay placement.  We first find an approximately optimal stab placement. Then we turn a stab placement into a hub placement within each cloud. Then a spanning tree on the clouds is found and ``Steinerised''.

Finding an optimal stab placement is a special case of the set cover problem. The maximum number of blobs pierced by a single stab is~$5$ (since this is the maximum number of unit disks that can have non-empty intersection while avoiding each other's centers). Thus, in this case the greedy heuristic for the set cover has an approximation ratio of $1+1/2+1/3+1/4+1/5=\rstabt$ \cite[Theorem~35.4]{cormen01introduction}.

Based on this approximation, a feasible hub placement $R_C$ within one cloud $C \in \myC$ can be obtained by applying Lemma~\ref{lem_stab2hub}; for this set of hubs it holds that
\[
    \s{R_C}
    \,\le\, \rhub \s{\optstab(\myB_C)} - 1.
\]
We can now interpret hubs $R_C$ as relays; if the hubs make the cloud $C$ connected, surely it holds for relays.

Let $R' = \bigcup_C{R_C}$ denote all relays placed this way. Since the blobs $\myB_C$ for different $C$ do not intersect, $\s{\optstab(\myB)} = \sum_C \s{\optstab(\myB_C)}$, so
\begin{equation}\label{eq_R1}
    \s{R'}
    \,\le\, \sum_C{\s{R_C}}
    \,\le\, \sum_C \left( \rhub \s{\optstab(\myB_C)} - 1 \right)
    \,=\, \rhub \s{\optstab(\myB)} - \s{\myC}.
\end{equation}

Next, we find $\MSTN(\myC)$ and place another set of relays, $R''$, along its edges.  Specifically, for each edge $e$ of the forest, we place $2$ relays at the endpoints of $e$, and  $\floor{\s{e}/r}$ relays every $r$ units starting from one of the endpoints. This ensures that all clouds communicate with each other; thus $R = R' \cup R''$ is a feasible solution. Since the number of edges in $\MSTN(\myC)$ is $\s{\myC}-1$,
\begin{equation}\label{eq_R2}
    \s{R''}
    \,=\, 2(\s{\myC} - 1) + \sum_e \left\lfloor \frac{\s{e}}{r} \right\rfloor
    \,<\, 2\s{\myC} +  \frac{\s{\MSTN(\myC)}}{r}.
\end{equation}
We obtain
\[
    \s{R}
    \,=\, \s{R'} + \s{R''}
    \,\le\, \left(\rhub + 1 + \frac{2}{\sqrt{3}}\right) \s{\optrel}
    \,<\, 6.73 \s{\optrel}
\]
from \eqref{eq_lower_bound_MStTN}--\eqref{eq_R2} and Lemma~\ref{lem_StRatio}.

\subsection{Running Time}

To implement the above algorithm in $O(n\log n)$ time, we construct the blobs (this can be done in $O(n\log n)$ time since the blobs are the union of disks centered on the sensors), assign each blob a unique colour, and initialise a Union-Find data structure for the colours. Next, we build the arrangement of the blobs, and sweep the arrangement $4$ times, once for each $d=5,4,3,2$; upon encountering a $d$-coloured cell of the arrangement, we place the stab anywhere in the cell, merge the corresponding $d$ colours, and continue. Finally, to place the hubs we do one additional sweep.

As for the last step -- building $\MSTN(\myC)$ -- it is easy to see that just as the ``usual'' minimum spanning tree of a set of points, $\MSTN(\myC)$ uses only the edges of the relative neighbourhood graph of the sensors (refer, e.g., to \citet[p.~217]{berg08computational} for the definition of the graph).
Indeed, let $pq$ be an edge of $\MSTN(\myC)$; let $p'$ and $q'$ be the sensors that are within distance $1$ of $p$ and $q$, respectively. If there existed a sensor $s'$ closer than $\myd{p'}{q'}$ to both $p'$ and $q'$, the edge $pq$ could have been swapped for a shorter edge (Figure~\ref{fig:nlogn}).
\begin{figure}[ht]
\centering
\scalebox{1.1}{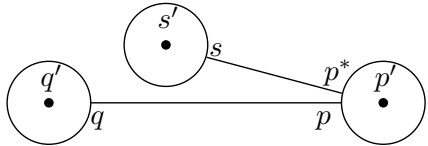}
\caption{Edge $pq$ could be swapped for a shorter edge.}\label{fig:nlogn}
\end{figure}

It remains to show how to build and sweep the arrangement of blobs in $O(n\log n)$ time.
Since the blobs are unions of unit disks, their total complexity is linear (see, e.g., \citet[Theorem~13.9]{berg08computational}). Moreover, the arrangement of the blobs also has only linear complexity (see Lemma~\ref{lem:arr} below); this follows from the fact that every point can belong to only a constant number of blobs (at most $5$). Thus, we can use sweep to build the arrangement in $O(n\log n)$ time, and also, of course, sweep the arrangement within the same time bound.

\begin{lemma}\label{lem:arr}
The arrangement of the blobs has linear complexity.
\end{lemma}
\begin{proof}
The vertices of the arrangement are of two types -- the vertices of the blobs themselves and the vertices of the intersection of two blobs (we assume that no three blobs intersect in a single point). The total number of the vertices of the first type is linear, so we focus on the vertices of the second type.

Let $A$ be a tile in the infinite unit-square tiling of the plane. There is not more than a constant number $K$ of blobs that intersect $A$ (since there is not more than a constant number of points that can be placed within distance $1$ from $A$ so that the distance between any two of the points is larger than~$1$). Let $n_i$ be the number of disks from blob $i$ that intersect $A$. Every vertex of the arrangement inside $A$ is on the boundary of the union of some two blobs. Because the union of blobs has linear complexity, the number of vertices that are due to intersection of blobs $i$ and $j$ is $O(n_i +n_j)$. Since there is at most $K$ blobs for which $n_i \ne 0$, we have
\[
    \sum_{i,j} (n_i+n_j) \le \binom{K}2 n(A),
\]
where $n(A)$ is the total number of disks intersecting $A$. Clearly, each unit disk intersects only a constant number of the unit-square tiles, and only a linear number of tiles is intersected by the blobs. Thus, summing over all tiles, we obtain that the total complexity of the arrangement is $O(K^2 n)=O(n)$.
\end{proof}


\section{A 3.11-Approximation Algorithm for\mytitlebreak One-Tier Relay Placement}\label{sec_apx1tierim}

In this section we first take care of clouds whose blobs can be stabbed with few relays, and then find an approximation to the hub placement by greedily placing the hubs themselves, without placing the stabs first, for the rest of the clouds. Together with a refined analysis, this gives a polynomial-time $3.11$-approximation algorithm.  We focus on nontrivial instances with more than one blob.

\subsection{Overview}

The basic steps of our algorithm are as follows:
\begin{enumerate}[noitemsep]
\item Compute optimal stabbings for the clouds that can be stabbed with few relays.
\item Connect the blobs in each of these clouds, using Lemma~\ref{lem_stab2hub}.
\item Greedily connect all blobs in each of the remaining clouds (``stitching'').
\item Greedily connect clouds into clusters, using 2 additional relays per cloud.
\item Connect the clusters by a spanning forest.
\end{enumerate}

Our algorithm constructs a set $A_r$ of ``red'' relays (for connecting blobs in a cloud, i.e., relays added in steps~1--3), a set $A_g$ of ``green'' relays (two per cloud, added in steps~4--5) and a set $A_y$ of ``yellow'' relays (outside of sensor range, added in step~5). Refer to Figures~\ref{fig:red} and~\ref{fig:greenyellow}. In the analysis, we compare an optimal solution $\optrel$ to our approximate one by subdividing the former into a set $\optrel_d$ of ``dark'' relays that are within reach of sensors, and into a set $\optrel_\ell$ of ``light'' relays that are outside of sensor range. We compare $\s{\optrel_d}$ with $\s{A_r}+\s{A_g}$, and $\s{\optrel_\ell}$ with $\s{A_y}$, showing in both cases that the ratio is less than~$3.11$.
\begin{figure}\centering
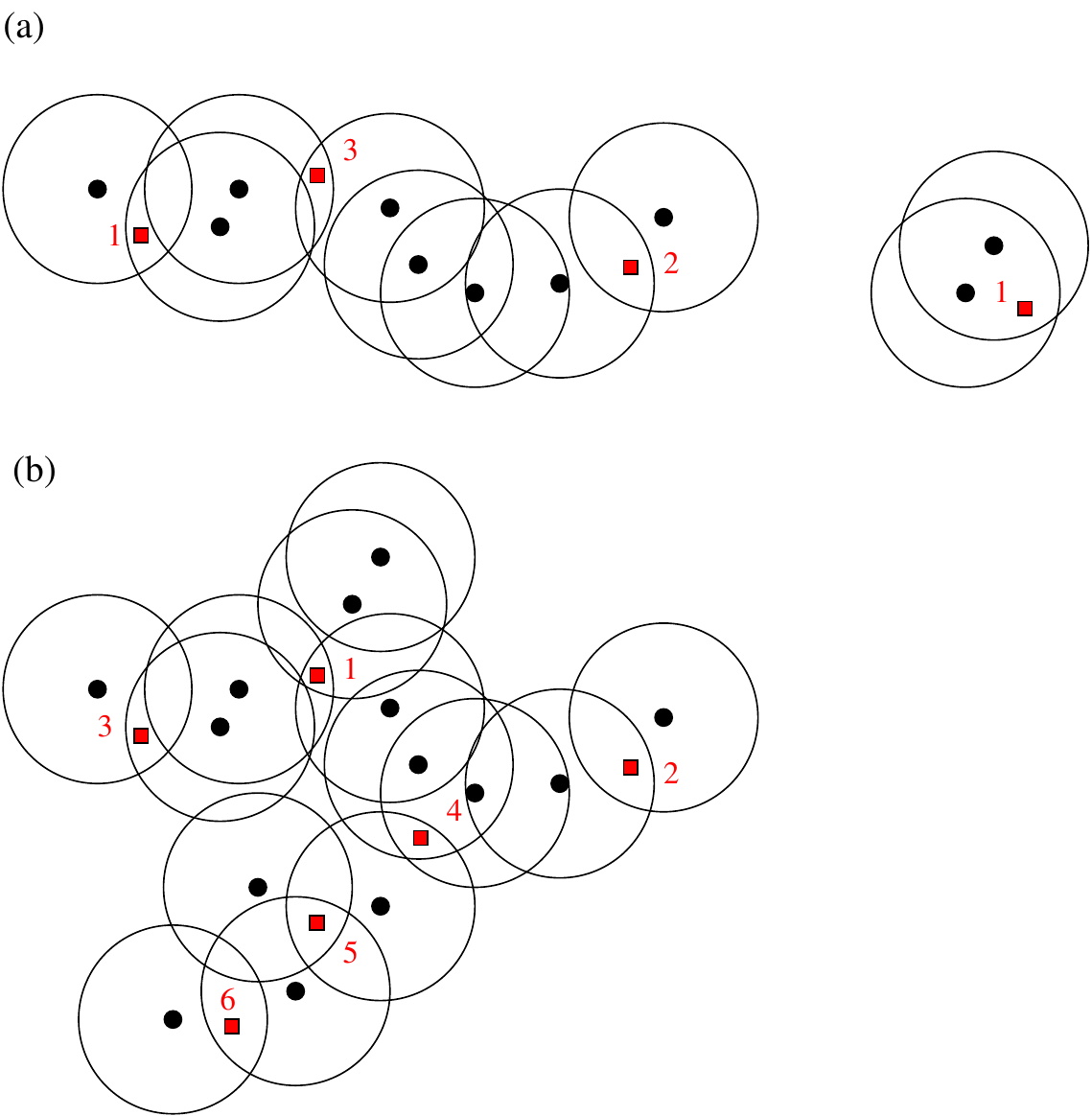
\caption{Red relays placed by our algorithm (the sensors are the solid circles); the numbers indicate the order in which the relays are placed within each cloud. (a)~Stab the clouds that can be stabbed by placing few relays; the clouds are then stitched by placing the hubs as in Lemma~\ref{lem_stab2hub}. (b)~Greedily stitch the other clouds.}\label{fig:red}
\end{figure}

\begin{figure}\centering
\includegraphics[page=1]{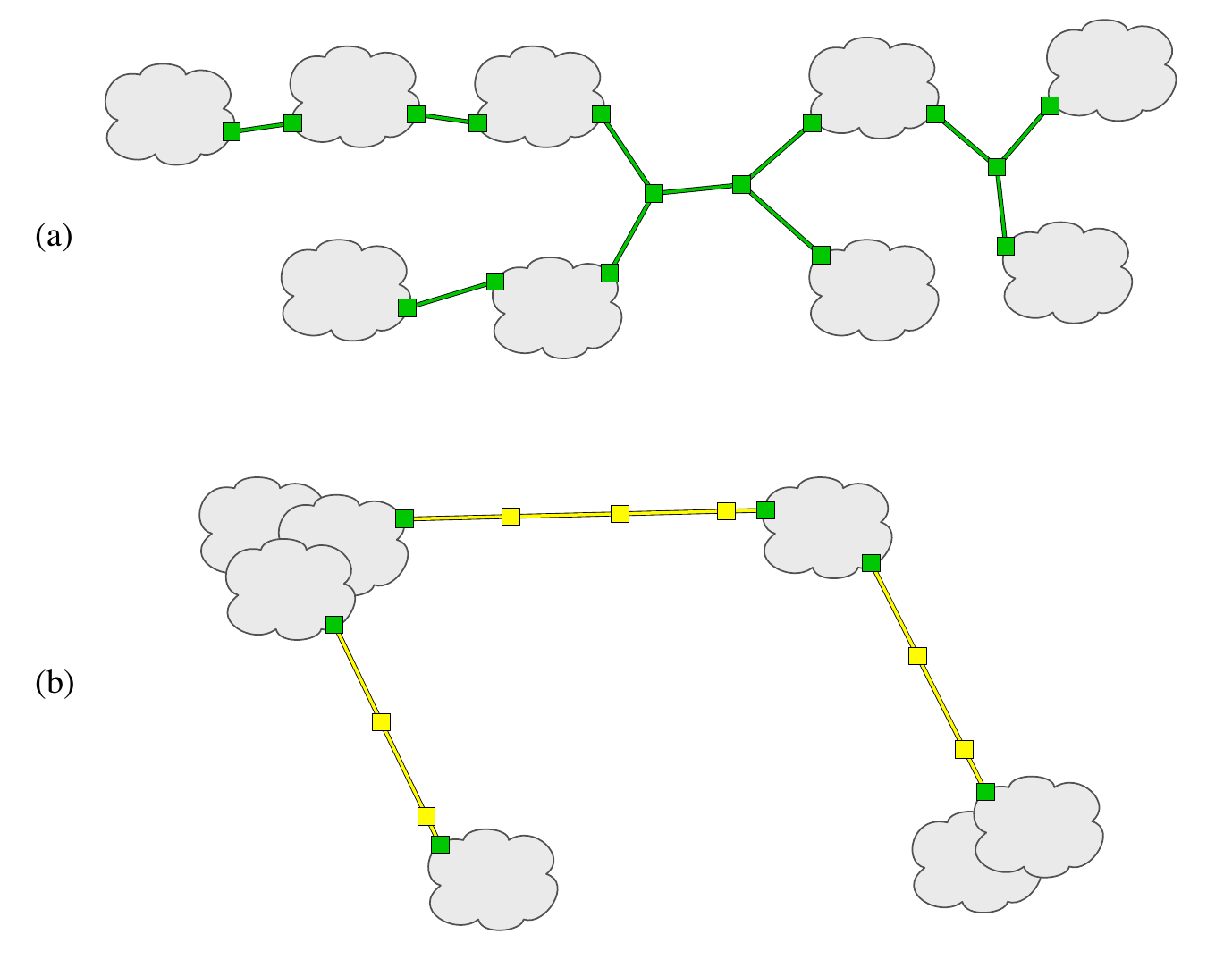}
\caption{(a)~Green relays connect clouds into clusters -- on average, we use at most 2 green relays per cloud. (b)~Green (inside clouds) and yellow (outside clouds) relays interconnect the cloud clusters by a spanning tree.}\label{fig:greenyellow}
\end{figure}

\subsection{Clouds with Few Stabs}

For any constant $k$, it is straightforward to check in polynomial time whether all blobs in a cloud $C\in\myC$ can be stabbed with $i<k$ stabs. (For any subset of $i$ cells of the arrangement of unit disks centered on the sensors in $C$, we can consider placing the relays in the cells and check whether this stabs all blobs.)  Using Lemma~\ref{lem_stab2hub}, we can connect all blobs in such a cloud with at most $2i-1$ red relays. We denote by $\myC^{i}$ the set of clouds where the minimum number of stabs is $i$, and by $\myC^{k+}$ the set of clouds that need at least $k$ stabs.

\subsection{Stitching a Cloud from \texorpdfstring{$\myC^{k+}$}{Ck+}}

We focus on one cloud $C \in \myC^{k+}$. For a point $y$ in the plane, let
\[
    \myB(y) = \{ B \in \myB_C : y \in B \}
\]
be the set of blobs that contain the point; obviously $\s{\myB(y)} \le 5$ for any $y$. For any subset of blobs $\myT \subseteq \myB_C$, define $ \myS(\myT,y) = \myB(y) \setminus \myT$ to be the set of blobs \emph{not from $\myT$} containing $y$, and define $V(\myT)$ to be the set of sensors that form the blobs in~$\myT$.

Within $C$, we place a set of red relays $A_r^C = \{ y_j : j = 1, 2, \dotsc \}$, as follows:
\begin{enumerate}
    \item Choose arbitrary $B_0 \in \myB_C$.
    \item Initialise $j \gets 1$, $\myT_j \gets \{ B_0 \}$.
    \item While $\myT_j \ne \myB_C$:
        \\[1ex]\hspace*{1em}%
        $\begin{aligned}
            y_j &\gets \arg \textstyle\max_{y} \{ \s{\myS(\myT_j,y)} : \myB(y) \cap \myT_j \ne \emptyset \}, \\
            \myS_j &\gets \myS(\myT_j,y_j), \\
            \myT_{j+1} &\gets \myT_j \cup \myS_j, \\
            j &\gets j + 1.
        \end{aligned}$
\end{enumerate}
That is, $y_j$ is a point contained in a maximum number of blobs \emph{not from $\myT_j$} that intersect a blob from $\myT_j$. In other words, we stitch the clouds greedily; the difference from the usual greedy (used in the previous section) is that we insist that {\em some} blob stabbed by $y_j$ is already in $\myT_j$.

By induction on $j$, after each iteration, there exists a path through sensors and/or relays between any pair of sensors in $V(\myT_j)$. By the definition of a cloud, there is a line segment of length at most $2$ that connects $V(\myT_j)$ to $V(\myB_C \setminus \myT_j)$; the midpoint of the segment is a location $y$ with $\myS(\myT_j,y) \ne \emptyset$. Since each iteration increases the size of $\myT_j$ by at least $1$, the algorithm terminates in at most $\s{\myB_C}-1$ iterations, and $\s{A_r^C} \le \s{\myB_C} - 1$.  The sets $\myS_j$ form a partition of $\myB_C \setminus \{ B_0 \}$.

We prove the following performance guarantee (the proof is similar to the analysis of greedy set cover.)

\begin{lemma}\label{lem:stitching}
    For each cloud\/ $C$ we have\/ $\s{A_r^C} \, \le \, 37 \s{\optrel_d \cap C} / 12 - 1$.
\end{lemma}

\begin{proof}
    For each $B \in \myB_C \setminus \{B_0\}$, define the weight $w(B) = 1 / \s{\myS_j}$, where $\myS_j$ is the unique set for which $B \in \myS_j$. We also set $w(B_0) = 1$. We have
    \begin{equation}\label{eq_J+1}
        \sum_{B\in \myB_C} \!\! w(B) \,=\, \s{A_r^C}+1.
    \end{equation}

    Consider a relay $z \in \optrel_d \cap C$, and find the smallest $\ell$ with $\myT_\ell \cap \myB(z) \ne \emptyset$, that is, $\ell = 1$ if $B_0 \in \myB(z)$, and otherwise $y_{\ell-1}$ is the first relay that pierced a blob from $\myB(z)$. Partition the set $\myB(z)$ into $\myU(z) = \myT_\ell \cap \myB(z)$ and $\myV(z) = \myB(z) \setminus \myU(z)$. Note that $\myV(z)$ may be empty, e.g., if $y_{\ell-1}=z$.

    First, we show that
    \[
        {\sum_{B \in \myU(z)} \!\!\! w(B)} \,\le\, 1 .
    \]
    We need to consider two cases. It may happen that $\ell = 1$, which means that $B_0 \in \myB(z)$ and $\myU(z) = \{B_0\}$. Then the total weight assigned to the blobs in $\myU(z)$ is, by definition, $1$. Otherwise $\ell > 1$ and $\myU(z) \subseteq S_{\ell-1}$, implying $ w(B) = 1/\s{S_{\ell-1}} \le 1/{\s{\myU(z)}}$ for each $B \in \myU(z)$.

    Second, we show that
    \[
        {\sum_{B \in \myV(z)} \!\!\! w(B)} \,\le\, \frac{1}{\s{\myV(z)}} + \frac{1}{\s{\myV(z)} - 1} + \dotsb + \frac{1}{1} .
    \]
    Indeed, at iterations $j \ge \ell$, the algorithm is able to consider placing the relay $y_j$ at the location $z$. Therefore $\s{\myS_j} \ge \s{\myS(\myT_j, z)}$. Furthermore,
    \[
        \myS(\myT_j, z) \setminus \myS(\myT_{j+1}, z)
        \,=\, \myB(z)\cap \myS_j
        \,=\, \myV(z) \cap \myS_j .
    \]
    Whenever placing the relay $y_j$ makes $\s{\myS(\myT_j, z)}$ decrease by a number $a$, exactly $a$ blobs of $\myV(z)$ get connected to $\myT_j$. Each of them is assigned the weight $w(C) \le 1/\s{\myS(\myT_j, z)}$.  Thus,
    \[
        \sum_{B \in \myV(z)} w(B)
        \,\le\, \frac{a_1}{a_1+a_2+\dotsb+a_n} + \frac{a_2}{a_2+a_3+\dotsb+a_n} + \dotsb + \frac{a_n}{a_n} ,
    \]
    where $a_1,a_2,\dotsc,a_n$ are the number of blobs from $\myV(z)$ that are pierced at different iterations, $\sum_i a_i = \s{\myV(z)}$.  The maximum value of the sum is attained when $a_1=a_2=\dotsb=a_n=1$ (i.e., every time $\s{\myV(z)}$ is decreased by 1, and there are $\s{\myV(z)}$ summands).

    Finally, since $\s{\myB(z)} \le 5$, and $\myU(z) \ne \emptyset$, we have $\s{\myV(z)}\le4$.  Thus,
    \begin{equation}\label{eq_Wz<}
        W(z) \,= {\sum_{B\in \myU(z)} \!\!\! w(B) \,+\, \sum_{B\in \myV(z)} \!\!\! w(B)}
        \,\le\, 1 + \frac14 + \frac13 + \frac12 + \frac11
        \,=\, \frac{37}{12}.
    \end{equation}
    The sets $\myB(z)$, $z \in \optrel_d \cap C$, form a cover of $\myB_C$. Therefore, from \eqref{eq_J+1} and \eqref{eq_Wz<},
    \[
            \frac{37}{12} \s{\optrel_d \cap C}
            \,\ge\! \sum_{z \in \optrel_d \cap C} \!\! W(z)
            \,\ge\! \sum_{B\in \myB_C} \!\! w(B) \,=\, \s{A_r^C} + 1. \qedhere
    \]
\end{proof}

\subsection{Green Relays and Cloud Clusters}

At any stage of the algorithm, we say that a set of clouds is \emph{interconnected} if, with the current placement of relays, the sensors in the clouds can communicate with each other.  Now, when all clouds have been stitched (so that the sensors within any one cloud can communicate), we proceed to interconnecting the clouds.  First we greedily form the collection of cloud \emph{clusters} (interconnected clouds) as follows. We start by assigning each cloud to its own cluster. Whenever it is possible to interconnect two clusters by placing one relay within each of the two clusters, we do so. These two relays are coloured green. After it is no longer possible to interconnect 2 clusters by placing just 2 relays, we repeatedly place 4 green relays wherever we can use them to interconnect clouds from 3 different clusters. Finally, we repeat this for 6 green relays that interconnect 4 clusters.

On average we place 2 green relays every time the number of connected components in the communication graph on sensors plus relays decreases by~one.

\subsection{Interconnecting the Clusters}

Now, when the sensors in each cloud and the clouds in each cluster are interconnected, we interconnect the clusters by a minimum Steiner forest with neighbourhoods. The forest is slightly different from the one used in the previous section. This time we are minimising the total number of relays that need to be placed along the edges of the forest in order to interconnect the clusters; we denote this forest by $\MSTN'$. The forest can be found by assigning appropriate weights to the edges of the graph on the clusters -- the weight of an edge is the number of relays that are necessary to interconnect two clusters.

After $\MSTN'$ is found, we place relays along edges of the forest just as we did in the simple algorithm from the previous section. This time though we assign colours to the relays. Specifically, for each edge $e$ of the forest, we place $2$ green relays at the endpoints of $e$, and $\floor{\s{e}/r}$ yellow relays every $r$ units starting from one of the endpoints.
As with interconnecting clouds into the clusters, when interconnecting the clusters we use 2 green relays each time the number of connected components of the communication graph decreases by one.  Thus, overall, we use at most $2\s{\myC}-2$ green relays.

\subsection{Analysis: Red and Green Relays}

Recall that for $i<k$, $\myC^i$ is the class of clouds that require precisely $i$ relays for stabbing, and $\myC^{k+}$ is the class of clouds that need at least $k$ relays for stabbing.  An optimal solution $\optrel$ therefore contains at least
\[
    \s{\optrel_d} \ge k\s{\myC^{k+}}+\sum_{i=1}^{k-1}i \s{\myC^i}
\]
dark relays (relays inside clouds, i.e., relays within reach of sensors).  Furthermore, $\s{\optrel_d \cap C} \ge 1$ for all $C$.

Our algorithm places at most $2i-1$ red relays per cloud in $\myC^i$, and not more than $37 \s{\optrel_d\cap C} / 12 -1$ red relays per cloud in $\myC^{k+}$. Adding a total of $2\s{\myC}-2$ green relays used for clouds interconnections, we get
\begin{align*}
    \s{A_r}+\s{A_g}
    &\,\le \sum_{C \in \myC^{k+}} \left( \frac{37}{12} \s{\optrel_d \cap C} - 1 \right) +
\sum_{i=1}^{k-1} (2i-1)\s{\myC^{i}} + 2 \s{\myC} - 2 \\
    &\,\le\, \frac{37}{12} \biggl( \s{\optrel_d} - \sum_{i=1}^{k-1} i\s{\myC^{i}} \biggr) + \s{\myC^{k+}} + \sum_{i=1}^{k-1} (2i+1)\s{\myC^{i}} - 2 \\
    &\,\le\, \frac{37}{12} \s{\optrel_d} + \s{\myC^{k+}}
    \,<\, \left( 3.084 +\frac{1}{k} \right) \s{\optrel_d}.
\end{align*}

\subsection{Analysis: Yellow Relays}

As in Section~\ref{ssec:lower-bounds}, let $\myR$ be the communication graph on the optimal set $\optrel$ of relays (without sensors). In $\myR$ there exists a forest $\myR'$ that makes the clusters interconnected.  Let $R'\subset\optrel$ be the relays that are vertices of $\myR'$. We partition $R'$ into ``black'' relays $\optrel_b = R' \cap \optrel_d$ and ``white'' relays $\optrel_w = R' \cap \optrel_\ell$ -- those inside and outside the clusters, respectively.

Two black relays cannot be adjacent in $\myR'$: if they are in the same cluster, the edge between them is redundant; if they are in different clusters, the distance between them must be larger than $r$, as otherwise our algorithm would have placed two green relays to interconnect the clusters into one. By a similar reasoning, there cannot be a white relay adjacent to 3 or more black relays in $\myR'$, and there cannot be a pair of adjacent white relays such that each of them is adjacent to 2 black relays. Refer to Figure~\ref{fig:forbid}. Finally, the maximum degree of a white relay is~5. Using these observations, we can prove the following lemma.
\begin{figure}\centering
\scalebox{1.2}{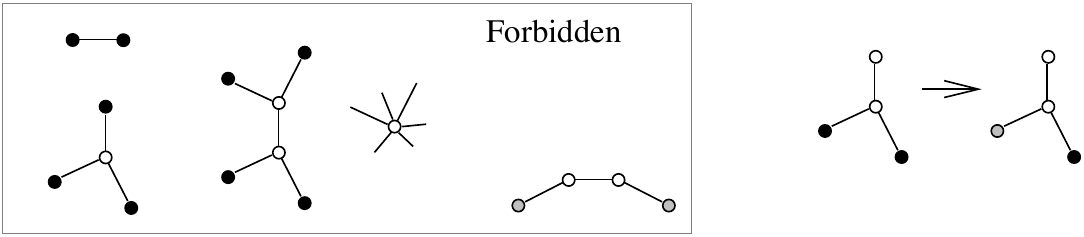}
\caption{Forbidden configurations and grey relays.}\label{fig:forbid}
\end{figure}

\begin{lemma}\label{lem_yellow}
    There is a spanning forest with neighbourhoods on cloud clusters that requires at most
    \[
        \left(\frac{4}{\sqrt{3}} + \frac{4}{5}\right) \s{\optrel_w} < 3.11 \s{\optrel_w}
    \]
    yellow relays on its edges.
\end{lemma}

\begin{proof}
    Let $\myD$ be the set of cloud clusters.  We partition $\myR'$ into edge-disjoint trees induced by maximal connected subsets of white relays and their adjacent black relays. It is enough to show that for each such tree $T$ that interconnects a subset of clusters $\myD' \subseteq \myD$, there is a spanning forest on $\myD'$ such that the number of yellow relays on its edges is at most $3.11$ times the number of white relays in $T$. As no pair of black relays is adjacent in $\myR'$, these edge-disjoint trees interconnect all clusters in $\myD$. The same holds for the spanning forests, and the lemma follows.

    Trees with only one white relay (and thus exactly two black relays) are trivial: the spanning forest needs only one edge with one yellow relay (and one green in each end). Therefore assume that $T$ contains at least two white relays.

    We introduce yet another colour. For each white relay with two black neighbours, arbitrarily choose one of the black relays and change it into a ``grey'' relay (Figure~\ref{fig:forbid}). Let $w$ be the number of white relays, let $b$ be the number of remaining black relays, and let $g$ be the number of grey relays in~$T$.

    First, we clearly have
    \begin{equation}\label{eq_yellow_bw}
        b \le w.
    \end{equation}
    Second, there is no grey--white--white--grey path, and each white relay is adjacent to another white relay. Therefore the ratio $(b+g)/w$ is at most $9/5$. To see this, let $w_2$ be the number of white relays with a grey and a black neighbour, let $w_1$ be the number of white relays with a black neighbour but no grey neighbour, and let $w_0$ be the number of white relays without a black neighbour. By degree bound, $w_2 \le 4 w_1 + 5 w_0 = 4 w_1 + 5 (w - w_2 - w_1)$; therefore $5w \ge 6 w_2 + w_1$. We also know that $w \ge w_2 + w_1$. Therefore
\begin{equation}\label{eq_yellow_95wbg}
    \frac95 w
    \,\ge\, \frac15 (6 w_2 + w_1) + \frac45 (w_2 + w_1)
    \,=\, (w_2 + w_1) + w_2
    \,=\, b + g.
\end{equation}
(The worst case is a star of $1+4$ white relays, $5$ black relays and $4$ grey relays.)

    Now consider the subtree induced by the black and white relays. It has fewer than $b+w$ edges, and the edge length is at most $r$. By Lemma~\ref{lem_StRatio}, there is a spanning forest on the black relays with total length less than ${(2/\sqrt{3})(b+w)r}$; thus we need fewer than ${(2/\sqrt{3})(b+w)}$ yellow relays on the edges.

    Now each pair of black relays in $T$ is connected. It is enough to connect each grey relay to the nearest black relay: the distance is at most $2$, and one yellow relay is enough. In summary, the total number of yellow relays is less than
    \[
        \begin{split}
        \frac{2}{\sqrt{3}} (b+w) + g
        &\,=\, \left(\frac{2}{\sqrt{3}} - 1\right) (b+w) + b+g + w \\
        &\,\le\, \left(\frac{2}{\sqrt{3}} - 1\right) 2 w + \frac{9}{5} w + w
        \,=\, \left(\frac{4}{\sqrt{3}} + \frac45 \right)w
        \,<\, 3.11 w .
        \end{split}
    \]
    The inequality follows from \eqref{eq_yellow_bw} and \eqref{eq_yellow_95wbg}.
\end{proof}

Thus, $\s{A_y} < 3.11 \s{\optrel_w} \le 3.11 \s{\optrel_\ell}$, and the overall approximation ratio of our algorithm is less than $3.11$.


\section{Inapproximability of One-Tier Relay Placement}\label{sec_inapx1tier}

We have improved the best known approximation ratio for one-tier relay placement from~7 to~3.11. A natural question to pose at this point is whether we could make the approximation ratio as close to 1 as we wish. In this section, we show that no PTAS exists, unless P${}={}$NP.
\begin{theorem}\label{thm:inapx}
    It is NP-hard to approximate one-tier relay placement within factor\/ $1 + 1/687$.
\end{theorem}

The reduction is from minimum vertex cover in graphs of bounded degree. Let $\myG = (V,E)$ be an instance of vertex cover; let $\Delta \le 5$ be the maximum degree of $\myG$. We construct an instance $\myI$ of the relay placement problem that has a feasible solution with\/ $k + 2\s{E} + 1$ relays if and only if $\myG$ has a vertex cover of size~$k$.

Figure~\ref{fig:inapx} illustrates the construction. Figure~\ref{fig:inapx}a shows the \emph{vertex gadget}; we have one such gadget for each vertex $v \in V$. Figure~\ref{fig:inapx}b shows the \emph{crossover gadget}; we have one such gadget for each edge $e \in E$. Small dots are sensors in the relay placement instance; each solid edge has length at most $1$. White boxes are \emph{good locations} for relays; there is one good location in each vertex gadget, and two good locations per crossover gadget. Dashed line shows a connection for relays in good locations in a crossover.

\begin{figure}[t]
    \centering
    \scalebox{0.9}{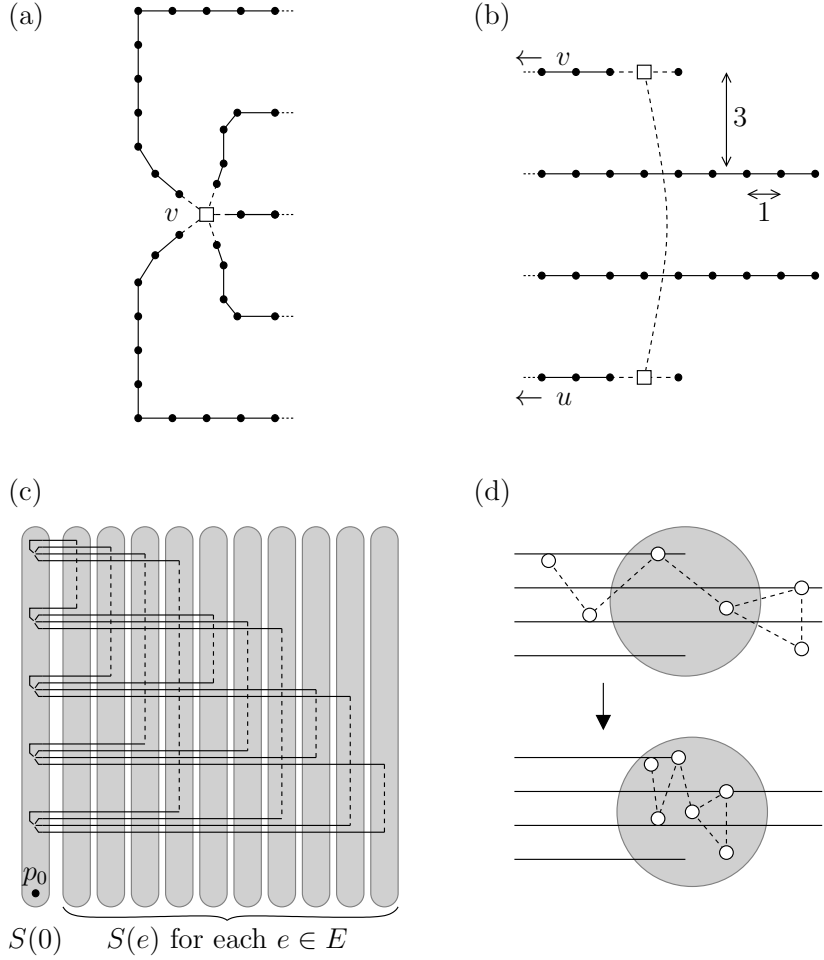}
    \caption{(a)~Vertex gadget for $v \in V$. (b)~Crossover gadget for $\{v,u\} \in E$. (c)~Reduction for $K_5$. (d)~Normalising a solution, step~1.}\label{fig:inapx}
\end{figure}

We set $r = 16(\s{V}+1)$, and we choose $\s{E}+1$ disks of diameter $r$ such that each pair of these disks is separated by a distance larger than $\s{V} r$ but at most $\poly(\s{V})$. One of the disks is called $S(0)$ and the rest are $S(e)$ for $e \in E$. All vertex gadgets and one isolated sensor, called $p_0$, are placed within disk $S(0)$. The crossover gadget for edge $e$ is placed within disk $S(e)$. There are noncrossing paths of sensors that connect the crossover gadget $e = \{u,v\} \in E$ to the vertex gadgets $u$ and $v$; all such paths (\emph{tentacles}) are separated by a distance at least~$3$. Good relay locations and $p_0$ cannot be closer than $1$ unit to a disk boundary.

Figure~\ref{fig:inapx}c is a schematic illustration of the overall construction in the case of $\myG = K_5$; the figure is highly condensed in $x$ direction. There are $11$ disks. Disk $S(0)$ contains one isolated sensor and $5$ vertex gadgets. Each disk $S(e)$ contains one crossover gadget. Outside these disks we have only parts of tentacles.

There are $4 \s{E} + 1$ blobs in $\myI$. The isolated sensor $p_0$ forms one blob. For each edge there are 4 blobs: two tentacles from vertex gadgets to the crossover gadget, and two isolated sensors in the crossover gadget.

Theorem~\ref{thm:inapx} now follows from the following two lemmata.

\begin{lemma}\label{lem:inapx-a}
    Let\/ $C$ be a vertex cover of\/ $\myG$. Then there is a feasible solution to relay placement problem\/ $\myI$ with\/ $\s{C} + 2\s{E} + 1$ relays.
\end{lemma}

\begin{proof}
    For each $v \in C$, place one relay at the good location of the vertex gadget~$v$. For each $e \in E$, place two relays at the good locations of the crossover gadget~$e$. Place one relay at the isolated sensor~$p_0$.
\end{proof}

\begin{lemma}\label{lem:inapx-b}
    Assume that there exists a feasible solution to relay placement problem\/ $\myI$ with\/ $k + 2\s{E} + 1$ relays. Then\/ $\myG$ has a vertex cover of size at most\/ $k$.
\end{lemma}

\begin{proof}
    If $k \ge \s{V}$, then the claim is trivial: $C = V$ is a vertex cover of size at most $k$. We therefore focus on the case $k < \s{V}$.

    Let $R$ be a solution with $k + 2\s{E} + 1$ relays. We transform the solution into a canonical form $R'$ of the same size and with the following additional constraints: there is a subset $C \subseteq V$ such that at least one relay is placed at the good relay location of each vertex gadget $v \in C$; two relays are placed at the good locations of each crossover gadget; one relay is placed at $p_0$; and there are no other relays. If $R'$ is a feasible solution, then $C$ is a vertex cover of $\myG$ with $\s{C} \le k$.

    Now we show how to construct the canonical form $R'$. We observe that there are $2\s{E}+1$ isolated sensors in $\myI$: sensor $p_0$ and two sensors for each crossover gadget. In the feasible solution $R$, for each isolated sensor $p$, we can always identify one relay within distance $1$ from $p$ (if there are several relays, pick one arbitrarily). These relays are called \emph{bound relays}. The remaining $k < \s{V}$ relays are called \emph{free relays}.

    \emph{Step~1.} Consider the communication graph formed by the sensors in $\myI$ and the relays $R$. Since each pair of disks $S(i)$, $i \in \{0\} \cup E$, is separated by a distance larger than $\s{V} r$, we know that there is no path that extends from one disk to another and consists of at most $k$ free relays (and possibly one bound relay in each end). Therefore we can shift each connected set of relays so that it is located within one disk (see Figure~\ref{fig:inapx}d). While doing so, we do not break any relay--relay links: all relays within the same disk can communicate with each other. We can also maintain each relay--blob link intact.

    \emph{Step~2.} Now we have a clique formed by a set of relays within each disk $S(i)$, there are no other relays, and the network is connected. We move the bound relay in $S(0)$ so that it is located exactly on $p_0$. For each $e \in E$, we move the bound relays in $S(e)$ so that they are located exactly on the good relay locations. Finally, any free relays in $S(0)$ can be moved to a good relay location of a suitable vertex gadget. These changes may introduce new relay--blob links but they do not break any existing relay--blob or relay--relay links.

    \emph{Step~3.} What remains is that some disks $S(e)$, $e \in E$, may contain free relays. Let $x$ be one of these relays. If $x$ can be removed without breaking connectivity, we can move $x$ to the good relay location of any vertex gadget. Otherwise $x$ is adjacent to exactly one blob of sensors, and removing it breaks the network into two connected components: component~$A$, which contains $p_0$, and component~$B$. Now we simply pick a vertex $v \in V$ such that the vertex gadget $v$ contains sensors from component $B$, and we move $x$ to the good relay location of this vertex gadget; this ensures connectivity between $p_0$ and $B$.
\end{proof}

\begin{proof}[Proof of Theorem~\ref{thm:inapx}.]
Let $\Delta, A, B, C \in \mathbb{N}$, with $\Delta \le 5$ and $C > B$. Assume that there is a factor
\[
    \alpha \,=\, 1 + \frac{C-B}{B + \Delta A + 1}
\]
approximation algorithm $\myA$ for relay placement. We show how to use $\myA$ to solve the following \emph{gap-vertex-cover} problem for some $0 < \eps < 1/2$: given a graph $\myG$ with $A n$ nodes and maximum degree $\Delta$, decide whether the minimum vertex cover of $\myG$ is smaller than $(B+\eps)n$ or larger than $(C-\eps)n$.

If $n < 2$, the claim is trivial. Otherwise we can choose a positive constant $\eps$ such that
\[
    \alpha - 1 \,<\, \frac{C-B-2\eps}{B +\eps + \Delta A+1/n}
\]
for any $n \ge 2$. Construct the relay placement instance $\myI$ as described above.

If minimum vertex cover of $\myG$ is smaller than $(B+\eps)n$, then by Lemma~\ref{lem:inapx-a}, the algorithm $\myA$ returns a solution with at most
$b = {\alpha ((B+\eps)n + 2\s{E} + 1)}$
relays. If minimum vertex cover of $\myG$ is larger than $(C-\eps)n$, then by Lemma~\ref{lem:inapx-b}, the algorithm $\myA$ returns a solution with at least
$c = (C-\eps)n + 2\s{E} + 1$
relays. As $2\s{E} \le \Delta A n$, we have
\[
    \begin{split}
    c - b
    &\,\ge\, (C-\eps)n + 2\s{E} + 1 - \alpha \bigl((B+\eps)n + 2\s{E} + 1\bigr) \\
    &\,\ge\, \bigl(C - B - 2 \eps - (\alpha-1)(B+ \eps + \Delta A + 1/n)\bigr) n
    \,>\, 0,
    \end{split}
\]
which shows that we can solve the gap-vertex-cover problem in polynomial time.

For $\Delta = 4$, $A = 152$, $B = 78$, $C = 79$, and any $0 < \eps < 1/2$, the gap-vertex-cover problem is NP-hard \cite[Theorem~3]{berman99some}.
\end{proof}

\begin{remark}
    We remind the reader that throughout this work we assume that radius $r$ is part of the problem instance. Our proof of Theorem~\ref{thm:inapx} heavily relies on this fact; in our reduction, $r = \Theta(\s{V})$. It is an open question whether one-tier relay placement admits a PTAS for a small, e.g., constant,~$r$.
\end{remark}


\section{A PTAS for Two-Tier Relay Placement}\label{sec_ptas}

In the previous sections we studied one-tier relay placement, in which
the sensors could communicate with each other, as well as with the
relays.  We gave a 3.11-approximation algorithm, and showed that the
problem admits no PTAS (for general $r$).  In this section we turn to
the two-tier version, in which the sensors cannot communicate with
each other, but only with relays.

The two-tier relay placement problem asks that we determine a set $R$
of relays such that there exists a tree $T$ whose internal nodes are
the set $R$ and whose leaves are the $n$ input points (sensors) $V$,
with every edge of $T$ between two relays having length at most $r$
and every edge of $T$ between a relay and a leaf (sensor) having
length at most 1.

We give a PTAS for this version of the problem, summarized in the
following theorem.

\begin{theorem}\label{thm:ptas}
The two-tier relay placement problem has a PTAS.
\end{theorem}

We give an overview of the method here; details and proofs appear in
the Appendix. Let $m$ be a (sufficiently large) positive integer constant; we will give a $(1+O(1/m))$-approximate solution. We distinguish between two cases: the \emph{sparse} case in which $\diam(V) \ge mnr$,
and the \emph{dense} case, in which $\diam(V) < mnr$.

In the sparse case, a solution can consist of long chains of relays,
with a number of relays not bounded by a polynomial in $n$; thus, we
output a succinct representation of such chains, specifying the
endpoints (which come from a regular grid of candidate locations).
The algorithm, then, is a straightforward reduction to the Euclidean
minimum Steiner tree problem.  See Appendix~\ref{sec_ptas_sparse}.

In the dense case, we compute and output an explicit solution.  In
this case, the set of possible locations of relays that we need to
consider is potentially large (but polynomial); we employ an
``iterated circle arrangement'' to construct the set, $G$, of
candidate locations.  Analysis of this set of candidates is done in
Appendix~\ref{sec_lem_rounding}, where we prove the structure lemma,
Lemma~\ref{lem_rounding}.  Armed with a discrete candidate set, we
then employ the $m$-guillotine method \cite{mitchell99guillotine} to
give a PTAS for computing a set of relays (a subset of $G$) that is
within factor $1+O(1/m)$ of being a minimum-cardinality set.  
The main idea is to optimize over the class of ``m-guillotine solutions'', which can be done using dynamic programming. An $m$-guillotine solution has a recursive property determined by ``guillotine cuts'' of the bounding box of the optimal solution (axis-parallel cuts of constant ($O(m)$) description complexity). We prove that an optimal solution that uses $k^*$ relays can be augmented with a set of $O(k^*/m)$ additional relays so that it has the $m$-guillotine property.


\section{Discussion}

In Section~\ref{sec_apx1tier} we presented a simple $O(n\log n)$-time 6.73-approximation algorithm for the one-tier relay placement.  If one is willing to spend more time finding the approximation to the set cover, one may use the semi-local optimisation framework of \citet{duh97approximation}, which provides an approximation ratio of $1+1/2+1/3+1/4+1/5-1/2$ for the set cover with at most 5 elements per set; hence we obtain a 5.73-approximation.

One can form a bipartite graph on the blobs and candidate stab locations as follows.  Pick a point within each maximal-depth cell of the arrangement of the blobs (maximal w.r.t\ the blobs that contain the cell); call these points ``red''.  Pick a point within each blob; call these points ``blue''.  Connect each blue point to the red points contained in the blob, represented by the blue point.  It is possible to pick the points so that the bipartite graph on the points is planar.  Then the stab placement is equivalent to the Planar Red/Blue Dominating Set Problem \cite{downey99parameterized} -- find fewest red vertices that dominate all blue ones.  We believe that the techniques of \citet{baker94approximation} can be used to give a PTAS for the problem.  Combined with the simple algorithm in Section~\ref{sec_apx1tier}, this would result in a $4.16$-approximation for the relay placement.

A more involved geometric argument may improve the analysis of yellow relays in Section~\ref{sec_apx1tierim}, bringing the constant $3.11$ in Lemma~\ref{lem_yellow} down to~$3$, which would improve the  approximation factor to $3.09$. Combining this with the possible PTAS for the Planar Red/Blue Dominating Set would yield an approximation factor of $3+\varepsilon$. We believe that a different, integrated method would be needed for getting below $3$: various steps in our estimates are tight with respect to~$3$. In particular, as the example in Figure~\ref{fig:worstcase} shows, our algorithm may find a solution with (almost) $3$ times more relays than the optimum.
\begin{figure}[h]
\centering
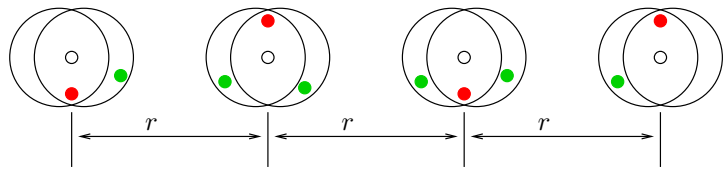
\caption{Unit disks centered on the sensors are shown. The optimum has $1$ relay per blob (hollow circles). Our algorithm may place $1$ red relay in every blob plus $2$ green relays in (almost every) blob.}\label{fig:worstcase}
\end{figure}


\section*{Acknowledgments}

We thank Guoliang Xue for suggesting the problem to us and for fruitful discussions, and Marja Hassinen for comments and discussions.  We thank the anonymous referees for their helpful suggestions.  

A preliminary version of this work appeared in the \emph{Proceedings of the 16th European Symposium on Algorithms} (ESA 2008) \cite{efrat08improved}.

Parts of this research were conducted at the Dagstuhl research center. AE is supported by NSF CAREER Grant 0348000. Work by SF was conducted as part of EU project  FRONTS (FP7, project 215270.)   JM is partially supported by grants from the National Science Foundation (CCF-0431030, CCF-0528209, CCF-0729019, CCF-1018388, CCF-1540890), NASA Ames, Metron Aviation, and Sandia National Labs.  JS is supported in part by the Academy of Finland grant 116547,  Helsinki Graduate School in Computer Science and Engineering (Hecse), and the Foundation of Nokia Corporation.  

\bibliographystyle{plainnat}
\bibliography{full}

\newpage
\appendix
\section{Appendix}

Here, we give details of the PTAS that was outlined in Section~\ref{sec_ptas}.

\subsection{Preliminaries}

Let $m$ be a (sufficiently large) positive integer constant. We present an algorithm that computes an approximate solution to the two-tier relay placement problem, with approximation factor $1+O(1/m)$ and running time that is polynomial in $n = \s{V}$. We focus on nontrivial instances in which at least $2$ relays are required (i.e., there is no single point that is within distance $1$ from all sensors).

We rephrase the statement of the two-tier relay placement problem:
Determine a set $R$ of relays such that there exists a tree $T$ whose
internal nodes are the set $R$ and whose leaves are the $n$ input
points (sensors) $V$, with every edge of $T$ between two relays having
length at most $r$ and every edge of $T$ between a relay and a leaf
(sensor) having length at most 1.  We refer to edges of length at most
$r$ between two relays as {\em blue} edges and edges of length at most
1 between a relay and a sensor as {\em red} edges; we refer to relays
that are incident on red edges as {\em red relays} and all other
relays as {\em blue relays}.  Our objective is to minimize $\s{R}$.

Let $D = \diam(V) - 1$. We say that an instance of the two-tier relay placement problem is \emph{sparse} (with respect to the constant $m$) if $D \ge mnr$, and \emph{dense} if $D < mnr$.  We treat sparse and dense instances with different techniques:
\begin{enumerate}[label=(\roman*)]
    \item A solution of a sparse instance typically consists of long chains of relays. The number of relays is not necessarily polynomial in $n$, and the algorithm must output a succinct representation that specifies the endpoints of such chains. However, the set of possible endpoints that we need to consider is relatively small, and we can use a simple regular grid to construct the set of candidate locations.
    \item A solution of a dense instance can be given explicitly. However, the set of possible locations of relays that we need to consider is large, and we must use an iterated circle arrangement to construct the set of candidate locations.
\end{enumerate}

We first show how to solve sparse instances in Section~\ref{sec_ptas_sparse}; the algorithm is a straightforward reduction to the Euclidean minimum Steiner tree problem.

\subsection{Sparse Instances}\label{sec_ptas_sparse}

We first discretize the problem by using a regular grid. Let $s = D/(nm)$; note that $s \ge r \ge 1$. Let $G$ be the regular square grid of spacing $s$ that covers the bounding rectangle of $V$; the number of points in $G$ is $O(n^2m^2)$.

We round the coordinates of each sensor $v \in V$ to the nearest grid point $G(v) \in G$; let $G(V) = \{ G(v) : v \in V \}$ be the set of the rounded sensor coordinates. Then we use a polynomial-time approximation algorithm \cite{arora98ptas,mitchell99guillotine} to find a Steiner tree $S$ for $G(V)$ with Steiner points in $G$ such that the total length $\s{S}$ of the edges is within factor ${1+1/m}$ of the smallest such tree. Without loss of generality, we can assume that $S$ has $O(n)$ edges.

Then we construct a feasible solution $R$ of the two-tier relay placement problem as follows:
\begin{enumerate}[label=(\roman*)]
    \item Replace each edge $\{a,b\}$ of $S$ by a chain of relays with spacing $r$ that connects $a$ to $b$.
    \item Connect each sensor $v \in V$ to the point $G(v)$ by a chain of relays.
\end{enumerate}

Clearly the solution $R$ is feasible and a compact representation of it can be constructed in polynomial time. Let us now analyze the approximation factor. To this end, consider an optimal solution $\optrel$ of the two-tier relay placement problem. Given $\optrel$, we can construct a Steiner tree $S^*$ for $V$ such that there are at most $n-2$ Steiner points. Each Steiner point is a relay of $\optrel$ with degree larger than $2$, and each $v \in V$ is a leaf node of $S^*$; moreover, $\s{S^*} \le n + \s{\optrel} r$.

If we round the coordinates of the vertices of $S^*$ to the nearest grid points in $G$, we obtain a Steiner tree $S_1$ for $G(V)$ with $\s{S_1} \le O(ns) + \s{\optrel} r$. By construction, $S$ is at most $1+1/m$ times larger than $S_1$; hence
\[
    \s{S} \,\le\, O(ns) + (1+1/m) \s{\optrel} r .
\]

In step~(i), we add at most $O(n) + \s{S}/r$ relays in total, and in step~(ii), we add $O(ns/r + n)$ relays in total. Hence,
\[
    \s{R} \,\le\, O(n + ns/r) + \s{S}/r \,\le\, O(ns/r) + (1+1/m) \s{\optrel}.
\]
Now observe that $\s{\optrel} \ge D/r$, as we must connect the pair of sensors that are at distance $\diam(V)$ from each other. Hence,
\[
    ns/r \,=\, D/(mr) \,\le\, (1/m) \s{\optrel},
\]
and we conclude that $\s{R}$ is within factor $1 + O(1/m)$ of $\s{\optrel}$.

\subsection{Dense Instances}

In Section~\ref{sec_lem_rounding} we prove the following lemma showing
that we can focus on a polynomial-size set $G$ of candidate relay
positions.

\begin{lemma}\label{lem_rounding}
For any fixed positive integer $m$ and a given dense instance of the
two-tier relay placement problem, we can construct in polynomial time
a (polynomial-size) set of points $G$ such that there exists a
feasible, ${(1+O(1/m))}$-approximate solution $R_1$ with $R_1
\subseteq G$.
\end{lemma}

Lemma~\ref{lem_rounding} implies that, in our quest for a PTAS for our
two-tier relay placement problem, it suffices for us to consider sets
of relays on the grid~$G$.

Our PTAS method is based on the $m$-guillotine framework of
\citet{mitchell99guillotine}.  We review a few definitions.  Let $A$ be
a connected set of line segments with endpoints in $G$. A {\em window}
$W$ is an axis-aligned rectangle whose defining coordinates are $x$- and
$y$-coordinates of the grid points $G$.  A {\em cut} is a horizontal
or vertical line $\ell$, through a point of $G$, that intersects
$\intW$, the interior of window $W$. The intersection, $\ell\cap
(A\cap \intW)$, of a cut $\ell$ with $A\cap \intW$ consists of a
discrete (possibly empty) set of crossing points where an edge from
$A$ crosses~$\ell$. Let the crossing points be denoted by
$p_1,\ldots,p_{\xi}$, in order along~$\ell$.  We define the {\em
  $m$-span}, $\sigma_m(\ell)$, of $\ell$ (with respect to $W$) to be
the empty set if $\xi\leq 2(m-1)$, and to be the (possibly
zero-length) line segment $p_{m}p_{\xi-m+1}$ otherwise.

A cut $\ell$ is {\em $m$-perfect with respect to $W$} if
$\sigma_m(\ell)\subseteq A$.  The edge set $A$ is {\em $m$-guillotine
  with respect to window $W$} if either
\begin{enumerate}
    \item $A\cap \intW=\emptyset$, or
    \item there exists an $m$-perfect cut, $\ell$, with respect to
      $W$, such that $A$ is $m$-guillotine with respect to windows
      $W\cap H^+$ and $W\cap H^-$, where $H^+$, $H^-$ are the closed
      halfplanes defined by~$\ell$.
\end{enumerate}
We say that $A$ is {\em $m$-guillotine} if $A$ is $m$-guillotine with
respect to the axis-aligned bounding box of~$A$.

\begin{figure}\centering
\includegraphics[page=2]{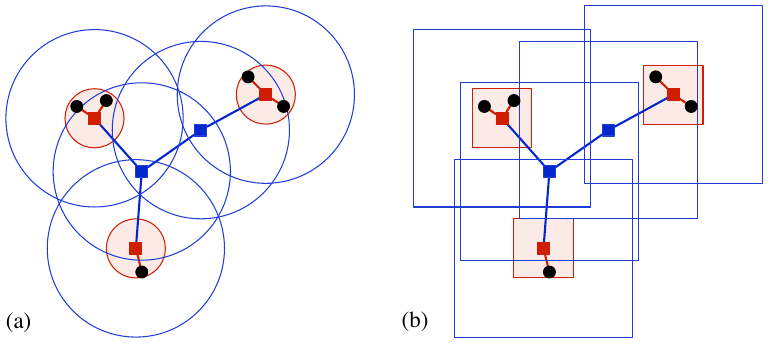}
\caption{Optimal solution $R_G^*$ and tree $T^*$. (a)~Red and blue disks. (b)~Red and blue squares.}\label{fig:redblue}
\end{figure}

Let $R_G^*$ denote an optimal solution to the two-tier relay placement
problem restricted to points of the grid $G$.  Associated with $R_G^*$
is a tree $T^*$ whose internal nodes are the set $R_G^*$ and whose
leaves are the $n$ input points (sensors) $V$, with every edge of
$T^*$ between two relays having length at most $r$ and every edge of
$T^*$ between a relay and a leaf (sensor) having length at most~1; see Figure~\ref{fig:redblue}.
Some edges of $T^*$ are blue (those of length at most $r$ between two
relays), and some edges are red (those of length at most 1 between a
red relay and a sensor).  Relays that are incident on red edges are
{\em red} and all other relays are {\em blue}.  We place {\em red
  disks} of radius 1 centered at each red relay and place {\em blue
  disks} of radius $r$ centered at {\em every} relay (blue or red).
Each red disk (resp., blue disk) has an associated bounding box, which
we call a {\em red square} (resp., {\em blue square}).  We observe
that each blue relay has constant degree in $T^*$ (in fact, at most 5,
based on the degree bound of Euclidean minimum spanning trees in the
plane), and no point in the plane lies in more than a constant number
of blue squares (for the same basic reasons as the degree bound --
otherwise, the set of relays could be reduced, while maintaining
communication connectivity).  Further, we observe that the union of
the edges bounding all blue squares is connected (since the blue edges
of $T^*$ form a subtree of $T^*$, and any two relays joined by a blue
edge must have their corresponding blue squares intersecting).

The $m$-guillotine method is typically applied to a set of {\em edges}
of some network.  We now define a related concept, that of a set of
relays being ``$m$-guillotine''.  Let $R\subseteq G$ be a set of
relays that is feasible for $V$, meaning that there is an associated
tree $T$ with leaves at $V$, blue edges of lengths at most $r$ and red
edges of lengths at most 1.  Let $E_{\blue}$ be the set of (axis-parallel)
edges bounding the blue squares associated with $R$, and let $E_{\red}$ be
the set of (axis-parallel) edges bounding the red squares associated
with $R$.  Consider a window $W$ (on the grid induced by the points
$G$) and a horizontal/vertical cut $\ell$ intersecting $\intW$.

The intersection, $\ell\cap (E_{\blue}\cap \intW)$, of a cut $\ell$ with
$E_{\blue}\cap \intW$ consists of a discrete (possibly empty) set of
crossing points where an edge from $E_{\blue}$ crosses~$\ell$.
Without loss of generality, consider a vertical cut $\ell$,
and let $ab=\ell\cap W$ denote its intersection with the window~$W$.
We now define the {\em blue $m$-span},
$\sigma_m^{\blue}(\ell)$, of $\ell$ (with respect to $W$), as follows.
As we walk along $ab$ from $a$ towards $b$, we enter various blue squares;
let $p_m$ be the $m$th entry point, if it exists; if we enter fewer than $m$
blue squares, then we define the blue $m$-span to be empty.
(Note that the point $a$ may lie within more than one blue square;
however, we can always assume that no point lies within more than a constant number of blue squares.)
Similarly, as we walk along $ab$ from $b$ towards $a$, we enter various blue squares;
let $q_m$ be the $m$th entry point, if it exists; if we enter fewer than $m$
blue squares, then we define the blue $m$-span to be empty.
Then, if $p_m$ is closer to $a$ than $q_m$, and the length of $p_mq_m$ is at least $2r$ (implying that
$p_mq_m$ fully crosses at least one blue square, entering and exiting it), we define the blue $m$-span,
$\sigma_m^{\blue}(\ell)$, of $\ell$ (with respect to $W$), to be the segment $p_mq_m$, a subsegment of $\ell$; otherwise,
we define the blue $m$-span to be empty. See Figure~\ref{fig:bluespan} for an illustration.
Note that, if nonempty, by definition the blue $m$-span has length at least $2r$, the side length of a blue square;
further, if the blue $m$-span is empty, then $ab$ intersects $O(m)$ blue squares.
We similarly
define the {\em red $m$-span}, $\sigma_m^{\red}(\ell)$, of $\ell$ (with
respect to $W$) based on the entry points along $\ell$ of red squares.
Note that, if nonempty, by definition the red $m$-span has length at least 1;
further, if the red $m$-span is empty, then $ab$ intersects $O(m)$ red squares.

\begin{figure}\centering
\includegraphics[page=3]{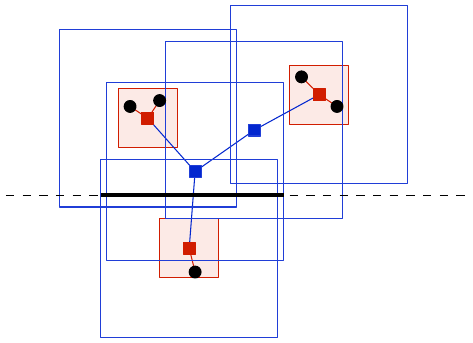}
\caption{A horizontal cut $\ell$ (dashed line) and the blue $m$-span $\sigma_m^{\blue}(\ell)$ for $m = 2$ (thick line).}\label{fig:bluespan}
\end{figure}

We say that a nonempty blue
$m$-span $\sigma_m^{\blue}(\ell)$ is {\em relay-dense}
if $R$ includes relays at
all points of $G_0^{\blue}$ that are corners of ($r/2$)-by-($r/2$) grid cells intersected
by  $\sigma_m^{\blue}(\ell)$; the corresponding relays are called a {\em relay-dense bridge}; see Figure~\ref{fig:densebridge}.  ($G$ includes the points $G_0^{\blue}$ of the regular square grid of spacing
$r/2$ that covers the bounding rectangle of $V$; see
Section~\ref{sec_lem_rounding}.)
Similarly, we say that a nonempty red $m$-span is {\em relay-dense} if
$R$ includes relays
at all points of $G_0^{\red}$ that are corners of (1/2)-by-(1/2) grid cells that contain points of $V$
(which implies that it includes at least one (red) relay within each disk of radius 1,
centered at $V$, that is intersected by $\sigma_m^{\red}(\ell)$);
the corresponding red relays are called a {\em (red) relay-dense bridge}.
($G$ includes the points $G_0^{\red}$ at the corners of the cells of the regular square grid of spacing 1/2 that contain points $V$;
see Section~\ref{sec_lem_rounding}.)

\begin{figure}\centering
\includegraphics[page=4]{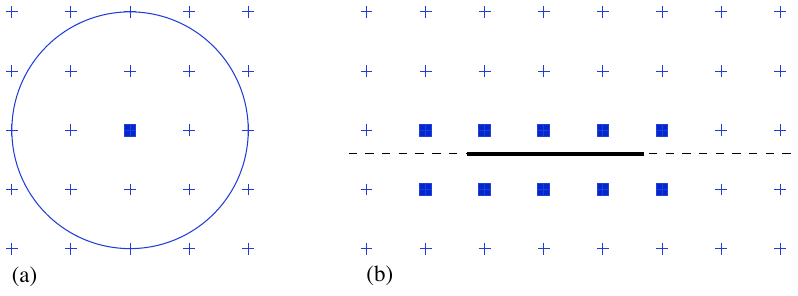}
\caption{(a)~Regular square grid $G_0^{\blue}$ of spacing $r/2$. (b)~A relay-dense $m$-span $\sigma_m^{\blue}(\ell)$.}\label{fig:densebridge}
\end{figure}

A cut $\ell$ is {\em $m$-perfect with respect to $R$} if
the nonempty $m$-spans (red and blue) along $\ell$ are relay-dense.
The set $R$ of relays is {\em $m$-guillotine
  with respect to window $W$} if either
\begin{enumerate}
    \item there are no blue squares of $R$ interior to $W$, or
    \item there exists an $m$-perfect cut, $\ell$, with respect to
      $W$, such that $R$ is $m$-guillotine with respect to windows
      $W\cap H^+$ and $W\cap H^-$, where $H^+$, $H^-$ are the closed
      halfplanes defined by~$\ell$.
\end{enumerate}
We say that $R$ is {\em $m$-guillotine} if $R$ is $m$-guillotine with
respect to the axis-aligned bounding box of~$R$.

\begin{lemma}\label{lem_guillotine_relays}
Let $m$ be a fixed positive integer, and let $V$ be a given dense
instance of the two-tier relay placement problem.  Then, for any set
$R'$ of relays that is feasible for the instance, there exists an
$m$-guillotine set $R$ of relays, also feasible for the instance, with
$\s{R}\leq (1+O(1/m))\s{R'}$.
\end{lemma}

\begin{proof}
We argue that we can add a small number, at most $O(1/m)\cdot\s{R'}$,
of relays to $R'$ to make it $m$-guillotine, if it is not already.

First, consider the set of blue and red squares associated with $R'$;
these are the bounding boxes of the disks of radius $r$ centered at
all relays and disks of radius 1 centered at red relays.  Let $E_{\blue}$
and $E_{\red}$ be the (axis-parallel) edges bounding the blue and red
squares, respectively.

By the standard $m$-guillotine arguments \cite{mitchell99guillotine},
we know that $E_{\blue}$ can be augmented by ``blue'' bridges ($m$-spans) of
length totalling $O(1/m)\cdot \s{E_{\blue}}$, so that the resulting set of
edges is $m$-guillotine (in the usual sense).  A similar statement
holds for the set $E_{\red}$ of red edges, which can be augmented by
``red'' bridges to make the set $m$-guillotine if it is not already
$m$-guillotine.  In fact, by a slight modification of the usual
argument, we can claim that we can augment $E_{\blue}\cup E_{\red}$ by a set of
red and blue bridges so that the resulting set is collectively
$m$-guillotine, with a common recursive partitioning by cuts: For each
cut $\ell$ in the recursive decomposition, we can choose $\ell$
(horizontal or vertical, in an appropriate position) so that the red
$m$-span (with respect to $E_{\red}$) and the blue $m$-span (with respect
to $E_{\blue}$) can each be charged off to the red and blue lengths of the
input, charging the input with a factor of $O(1/m)$.  In more detail,
we define the cost of a vertical cut at position $x$ to be
$f_{\red}(x)+(1/r)f_{\blue}(x)$, where $f_{\red}(x)$ is the length of the $m$-span at
position $x$ of the red edges $E_{\red}$, and $f_{\blue}(x)$ is the length of the
$m$-span at position $x$ of the blue edges $E_{\blue}$.  The union of
$m$-spans over all $x$ defines a ``red region'' and a ``blue region''
with respect to vertical cuts, and the integral, $\int_x
[f_{\red}(x)+(1/r)f_{\blue}(x)] dx$ gives the weighted sum of the areas of these
regions, weighting the length of the blue $m$-span by $1/r$ (this is
to reflect the fact that the blue squares are $r$ times as large as
the red squares, and our goal is to count the cardinality of squares).
Similarly, we can define red and blue regions with respect to {\em
  horizontal cuts}, and functions, $g_{\red}(y)$ and $g_{\blue}(y)$, that define
the cost of such cuts.

The (weighted) areas of the red/blue regions with respect to
horizontal cuts are given by the integrals $\int_y g_{\red}(y) dy$ and
$\int_y (1/r)g_{\blue}(y) dy$; equivalently they are given by the integrals,
$\int_y h_{\red}(x) dx$ and $\int_x (1/r)h_{\blue}(x) dx$, of the lengths
$h_{\red}(x)$ and $h_{\blue}(x)$ of the vertical sections of the regions.  These
lengths ($h_{\red}(x)$ and $h_{\blue}(x)$) represent the ``chargeable length'' of
vertical cuts at position $x$.  Then, assuming that the weighted sum
of the areas of the red and blue regions with respect to horizontal
cuts is greater than that with respect to vertical cuts, we get that
there must be a vertical cut at some position $x^*$ where
$f_{\red}(x^*)+(1/r)f_{\blue}(x^*) \leq h_{\red}(x^*) + (1/r)h_{\blue}(x^*)$.  A vertical
cut through $x^*$ can then have its red $m$-span charged to $O(1/m)$
of the length of edges bounding red squares (i.e., $O(1/m)$ times the
number of red squares) and its blue $m$-span charged to $O(1/mr)$
times the length of edges bounding blue squares (i.e., $O(1/m)$ times
the number of blue squares/relays).

Next, consider each blue bridge ($m$-span) that the above argument
shows we can afford to add (charging $O(1/m)$ times the number of
relays), to make $E_{\blue}$ $m$-guillotine (as a network).  We can convert
any such blue bridge of length $\lambda$ into $O(\ceil{\lambda/r})$
relays on or near the corresponding cut $\ell$, in order to make the
cut relay-dense (i.e., to include relays at all points of $G_0^{\blue}$ that
are corners of ($r/2$)-by-($r/2$) grid cells intersected by the blue
bridge).
Similarly, each red bridge ($m$-span), of length $\lambda$, that is added in order to make
$E_{\red}$ $m$-guillotine (as a network) can be converted into $O(\ceil{\lambda})$
red relays that stab all disks of radius 1, centered at points of $V$, that are
intersected by the red bridge.
Similarly, we say that a nonempty red $m$-span is {\em relay-dense} if
$R$ includes at least one (red) relay from the set $G$ within each disk of radius 1,
centered at $V$, that is intersected by $\sigma_m^{\red}(\ell)$.
The purpose of relay-dense bridges of red relays is to ensure that
sensors within distance 1 of the boundary of $W$ are covered by the unit disks
of the red relay-dense bridge, if they are not covered by
the unit disks associated with the $O(m)$ unbridged red relays.  This allows for ``separation''
between subproblems in the dynamic program.

Finally, we conclude that the total number of relays added is $O(1/m)$
times $\s{R}$.  If we let the lengths of blue bridges be
$\lambda_i^{\blue}$ and let the lengths of red bridges be
$\lambda_j^{\red}$, then we know, from the charging scheme for making
$E_{\blue}$ and $E_{\red}$ $m$-guillotine, that $(1/r)\sum_i \lambda_i^{\blue} +
\sum_j \lambda_j^{\red}$ is at most $O(1/mr)$ times the total
perimeters of all blue squares plus $O(1/m)$ times the total
perimeters of all red squares.  This implies that $(1/r)\sum_i
\lambda_i^{\blue} + \sum_j \lambda_j^{\red}$ is at most $O(1/m)$ times
the number of relays of $R$.  By construction, we know that if
$\lambda_i^{\blue}$ is nonzero, then $\lambda_i^{\blue}\geq r$, and that
if $\lambda_j^{\red}$ is nonzero, then $\lambda_j^{\red}\geq 1$.  Thus,
the total number of relays added, $O(\sum_i \ceil{\lambda_i^{\blue}/r} +
\sum_j \ceil{\lambda_j^{\red}})$ is at most $O(1/m)$ times the number
of relays of $R$.
\end{proof}

Applying Lemma~\ref{lem_guillotine_relays} and Lemma~\ref{lem_rounding} yields the following

\begin{corollary}
Let $m$ be a fixed positive integer, and let $V$ be a given dense instance of
the two-tier relay placement problem.
Then, there exists an $m$-guillotine
set $R\subseteq G$ of relays with $\s{R}\leq (1+O(1/m))\s{R^*}$.
\end{corollary}

We now describe the algorithm that yields the claimed PTAS for
computing $R^*$ approximately.  The algorithm is based on a dynamic
program to compute a minimum-cardinality $m$-guillotine set, $R$, of
relays that obeys certain connectivity and coverage constraints.  A
subproblem is specified by a window $W$ (with coordinates from the
grid $G$), together with a constant (depending on $m$) amount of boundary information:
\begin{enumerate}
\item A set of $O(m)$ (unbridged) blue relays (from $G$) whose
  blue squares intersect the boundary of $W$, and a set of $O(m)$ red
  relays (from $G$) whose red squares intersect the boundary of~$W$.
\item The blue $m$-spans, which determine the set of relays
  forming the relay-dense bridge; in order to describe the relay-dense
  bridge, we have only to specify the endpoints of the blue $m$-span,
  as this gives a succinct encoding of the relays that form the
  relay-dense bridge.  There are up to 4 blue $m$-spans, one per side
  of $W$.
\item The red $m$-spans, which determine the set of red relays
  forming the relay-dense bridge; in order to describe the relay-dense
  bridge, we have only to specify the endpoints of the red $m$-span,
  as this gives a succinct encoding of the relays that form the
  relay-dense bridge.
\item Connectivity requirements among the relays specified along
  the boundary.  Specifically, there are $O(m)$ entities specified
  along the boundary of $W$: $O(m)$ unbridged relays per side of
  $W$, and up to two $m$-spans per side of $W$, which encode the
  relay-dense bridges.  The connectivity requirements are specified as
  a collection of (disjoint) subsets of these $O(m)$ entities, with
  the understanding that each such subset of entities must be
  connected within $W$ by edges of the communication graph that
  connects two relays if and only if they are at distance at most $r$
  from each other.
\end{enumerate}
We require that the solution to a subproblem gives a set of red relays
whose unit disks cover all sensors interior to $W$ that are not
covered by the relays specified along the boundary (either the $O(m)$
unbridged red relays or the relay-dense bridges of red relays),
together with a set of blue relays within the subproblem that enable
connectivity.  The dynamic program optimizes over all choices of cuts
$\ell$ (horizontal or vertical) of $W$, and all choices of boundary
information along the cut $\ell$ that are compatible with the boundary
information for $W$.  Since, for fixed $m$, there are a polynomial
number of different subproblems, and a polynomial number of choices of
cuts and boundary information within the dynamic programming
recursion, we have completed the proof of the main result, Theorem~\ref{thm:ptas}.

\subsection{Proof of Lemma~\ref{lem_rounding}}\label{sec_lem_rounding}

\paragraph{\boldmath Construction of $G$.}

Let $G_0^{\blue}$ be the regular square grid of spacing $r/2$ that
covers the bounding rectangle of $V$; the points of $G_0^{\blue}$ are
the corners of ($r/2$)-by-($r/2$) squares that form a regular grid.
By construction, the grid $G_0^{\blue}$ has ${O((D/r)^2 + 1)} = O(n^2 m^2)$ points.

Consider now the regular square grid of spacing $1/2$ that covers the
bounding rectangle of~$V$.  Let $G_0^{\red}$ be the subset of these
grid points that are within distance 1 of one of the $n$ input points~$V$.
Then, $G_0^{\red}$ are the corners of ($1/2$)-by-($1/2$) squares
of a regular grid, and there are only $O(n)$ such corner points.

We construct an \emph{iterated circle arrangement} $C_0, C_1, \dotsc, C_m$ recursively as follows. As the base case, let $C_0 = G_0^{\blue}\cup V$. Then, for each $i = 1, 2, \dotsc, m$, the arrangement $C_i$ is constructed as follows: For each point $p \in C_{i-1}$, add $p$ to $C_i$, as well as the \emph{special points} located $1$ unit below $p$ and $r$ units below $p$. Then draw circles of radii $1,r,r+1,2r,2r+1,3r,3r+1,\dotsc,mr,mr+1$ centered at each point of $C_{i-1}$. Finally, add the vertices (crossing points) of this arrangement to~$C_i$.

Finally, let $G = C_m\cup G_0^{\red}$. The size of $G$ is polynomial in $n$, for a constant $m$.

\paragraph{\boldmath Construction of $R_1$.}

Consider an optimal solution $\optrel$ and the minimum-length communication tree $T^*$ associated with it; let $B^*$ be the blue subtree of $T^*$. Since $B^*$ is an MST of $\optrel$, the maximum degree of any relay in $B^*$ is $5$. Hence there is a set $E'$ of $O(\s{R^*}/m)$ blue edges whose removal breaks $B^*$ into subtrees, each having at most $m$ relays.

We construct a new solution $R_1$ that conforms to $G$ as follows; see Figure~\ref{fig:snapdense} for an illustration. First, for each edge $\{x,y\} \in E'$, we add a new relay $z$ at a point of $G_0^{\blue}$ that is within distance $r$ from $x$ and $y$; then we replace the edge $\{x,y\}$ by two edges, $\{x,z\}$ and $\{z,y\}$. The new blue edge $\{x,z\}$ becomes part of the same subtree as $x$, and the edge $\{z,y\}$ becomes part of the same subtree as $y$.

\begin{figure}
    \centering
    \scalebox{1.0}{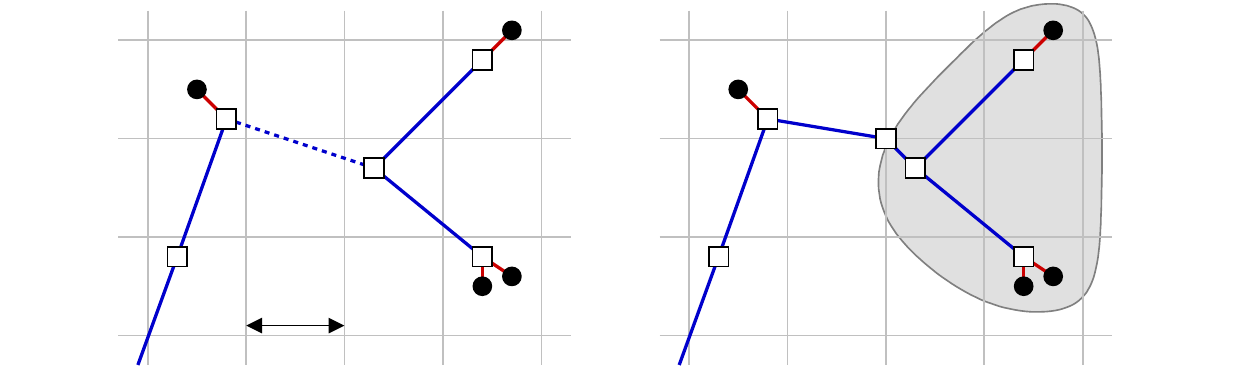}
    \caption{Construction of $R_1$ for dense instances. Each edge in the set $E'$ is replaced by a path of length $2$. The grey area shows one of the subtrees. The leaf nodes of the subtree are now pinned at points of $C_0$, and there are at most $m$ internal nodes. Figure~\ref{fig:pinsubtree} shows how we can ``pin'' the internal nodes of the subtree to points of $G$.}\label{fig:snapdense}
\end{figure}

So far we have added $O(\s{R^*}/m)$ new relays, and we have partitioned the communication tree into edge-disjoint subtrees. The leaf nodes of the subtrees are sensors or newly added relays; both of them are located at points of $C_0$. There are at most $m$ internal nodes in each subtree. To prove Lemma~\ref{lem_rounding}, it is sufficient to show that we can move the internal relays so that they are located at points of $G$, and while doing so we do not break any sensor--relay or relay--relay connections.

We say that a relay is \emph{pinned} if it is located at a point of $G$. We show how to pin all relays in $m$ iterations. In the beginning, we have partitioned the communication tree into edge-disjoint subtrees where the leaf nodes are pinned at $C_0$ and there are at most $m$ internal nodes in each subtree. In what follows, we will show that we can move the internal nodes and refine the partition so that after iteration $i$, we have partitioned the communication tree into edge-disjoint subtrees where the leaf nodes are pinned at $C_i$ and there are at most $m-i$ internal nodes in each subtree. Hence after $m$ iterations, we have pinned all nodes at $C_m = G$.

Observe that it is sufficient to show that if $T$ is a subtree where leaf nodes are pinned at $C_i$, then we can move the internal nodes so that at least one node $x$ becomes pinned at $C_{i+1}$. Then we can partition the subtree $T$ into smaller, edge-disjoint subtrees that share the leaf node $x$. Each of the new subtrees has all leaf nodes pinned at $C_{i+1}$ and they also have fewer internal nodes than $T$.

Consider a subtree $T$. Start translating (in any direction) the internal nodes of $T$, while keeping the leaf nodes in place. This translation will cause some edge $e$ linking a leaf node $p$ to an internal node $u$ to reach its maximum length ($r$ for blue edges or $1$ for red edges). At this moment, the point $u$ lies on a circle $c$ of radius $r$ or $1$ centered at~$p$; see Figure~\ref{fig:pinsubtree} for an illustration.

\begin{figure}[b!]
    \centering
    \scalebox{1.0}{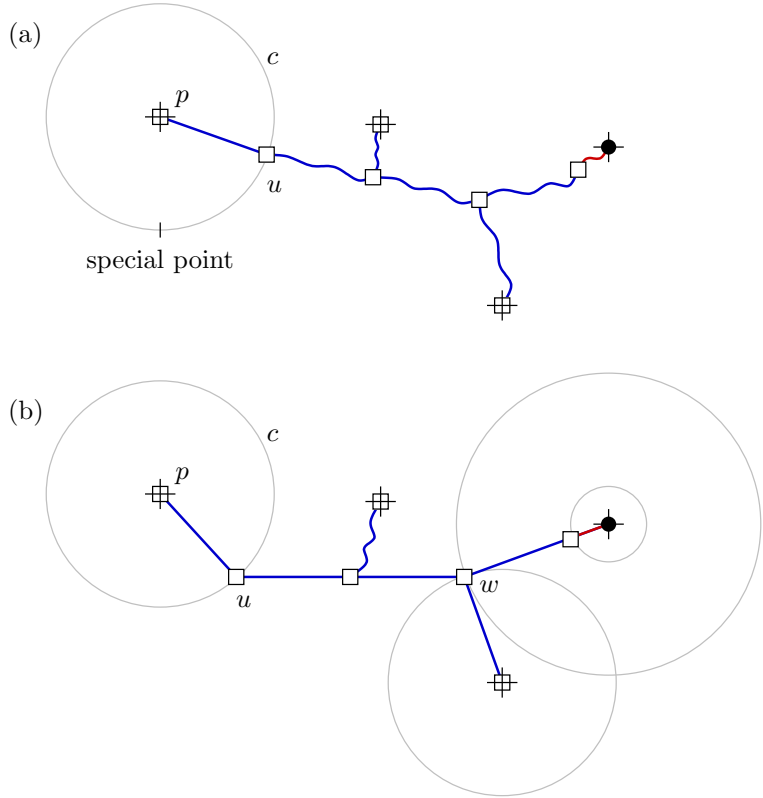}
    \caption{Pinning a subtree. Black dots are sensors and white boxes are relays. The leaf nodes are pinned at points of $C_i$. Straight lines denote ``taut'' edges, i.e., red edges of length $1$ or blue edges of length $r$; other edges are shown with winding lines. (a)~We have translated the internal nodes until the edge between a leaf node $p$ and an internal node $u$ becomes taut. Then we try to slide $u$ along circle $c$ towards the special point. If we succeed, we have pinned $u$ at the special point, which is in $C_{i+1}$. (b)~In this case we cannot slide $u$ to the special point. However, from the subtree of taut edges, we can find a node $w$ that is now located at a point of $C_{i+1}$ -- in this case at the intersection of a circle of radius $r+1$ centered at a sensor, and a circle of radius $r$ centered at a pinned relay.}\label{fig:pinsubtree}
\end{figure}

Now slide the point $u$ along circle $c$ until some other constraint becomes binding. We now do not consider $T$ to be a rigid structure; rather, we allow the internal nodes to move arbitrarily, subject to the upper bound constraints on each edge. The edges serve as ``cables'', which readily shorten, but cannot extend beyond their maximum length.

If we can slide $u$ along $c$ to the special point located below $p$, then we have pinned $u$ at a point of $C_{i+1}$. Otherwise some edges of $T$ reach their upper bound of length while we slide $u$ along $c$; i.e., one or more cables becomes taut. At this moment, some subset of the edges of $T$ are taut (at their length upper bound); this subset of edges forms a tree, $T'$. There are two cases:
\begin{enumerate}[label=(\roman*)]
    \item $T'$ is a path. Then $T'$ is a path linking $u$ to a leaf node $p'$ of $T$. Further, $T'$ is a \emph{straight} path of edges, each at its upper length bound; in fact, all edges of $T'$ will be of length $r$, the upper bound for relay--relay edges, except, possibly, the edge incident on $p'$ (if $p'$ is a sensor). Thus, in this position, $u$ lies at the intersection of circle $c$ with a circle centered at a point of $P$ of radius $jr$ or $jr+1$, for some integer~$j \le m$. As we had $p' \in C_i$, we will have now $u \in C_{i+1}$.
    \item $T'$ is not a path. Then, $T'$ has some node, $w$ (possibly equal to $u$) that is connected to leaves of $T$ by two (or more) straight paths, each of length $jr$ or $jr+1$, for some integer $j \le m$; see Figure~\ref{fig:pinsubtree}b. This implies that $w$ lies at a point of $C_{i+1}$.
\end{enumerate}

In both cases we have moved the internal nodes of $T$ so that one of them becomes pinned at $C_{i+1}$; moreover, none of the communication links are broken, leaf nodes are held in place, and no new relays are added.

\end{document}